%% file: ms.tex
\documentclass[preprint]{elsarticle}
\usepackage[english]{babel}

\usepackage[a4paper,top=3cm,bottom=2cm,left=3cm,right=3cm,marginparwidth=1.75cm]{geometry}

\usepackage{amsmath}
\usepackage{amsthm}
\usepackage{graphicx}
\usepackage[colorinlistoftodos]{todonotes}
\usepackage[colorlinks=true, allcolors=blue]{hyperref}
\usepackage{xcolor}
\usepackage{cleveref}
\usepackage{mathtools}

\newcommand{\shift}{\texttt{SHIFT}}
\newtheorem{theorem}{Theorem}
\newtheorem{lemma}[theorem]{Lemma}
\newtheorem{corollary}[theorem]{Corollary}
\newtheorem{definition}{Definition}
\newtheorem{observation}{Observation}
\newcommand{\calS}{\mathcal{S}}

\newcommand{\lcrv}{{\mathrm {LCV}}}
\newcommand{\rcrv}{{\mathrm {RCV}}}
\newcommand{\lue}{{\mathrm {LUE}}}
\newcommand{\rue}{{\mathrm {RUE}}}
\newcommand{\OPT}{{\mathrm OPT}}

\newcommand{\yjyx}{{y\,:\, x_j \le y \le x_n}}
\newcommand{\yxyi}{{y\,:\, x_0 \le y \le x_i}}
\newcommand{\yiyj}{{y\,:\, x_i \le y \le x_j}}
\newcommand{\tjtx}{{t\,:\, x_j \le x_t \le x}}
\newcommand{\txti}{{t\,:\, x\le x_t \le x_i}}

\DeclareMathOperator*{\argmax}{arg\,max}

\title{Minmax Regret for Sink Location on Dynamic Flow Paths with General Capacities 
}

\author[HKUST]{Mordecai Golin \fnref{fn1}}
\address[HKUST]{Hong Kong  UST.  golin@cse.ust.hk}

\author[CMU]{Sai Sandeep}
\address[CMU]{Carnegie Mellon University. spallerl@andrew.cmu.edu}
\fntext[fn1]{Work partially supported by Hong Kong RGC CERG  Grant	16208415}

\begin{document}

\begin{abstract}
In dynamic flow networks, every vertex starts with items (flow) that need to be shipped to designated sinks. 
All edges have two associated quantities: length,  the amount of time required  for a particle to traverse the edge,   and capacity,  the number of units of flow that can enter the edge in unit time.
The goal is to move all the flow to the sinks. A  variation of the problem, modelling evacuation protocols, is to {\em find the sink location(s)}  that minimize  evacuation time, restricting the flow to be {\em confluent}. Solving this problem is NP-hard on general graphs, and thus research into optimal algorithms has  traditionally been restricted to  special graphs such as paths, and trees. 

A specialized version of robust optimization 
 is minmax {\em regret}, in which the  input flows at the vertices are only partially defined by constraints. The goal is to find a  sink location that has the 
 minimum{ regret} over all the input flows that satisfy the partially defined constraints.  Regret for a fully defined input flow and  a sink  is defined to be the difference between the  evacuation time to that sink  and the optimal evacuation time. 
 
A large recent literature  derives polynomial time algorithms for the minmax regret $k$-sink location problem on paths and trees {\em under the simplifying condition that all edges have the same (uniform) capacity}.
 This paper develops  a $O(n^4 \log n)$  time algorithm for  the minmax regret $1$-sink problem on paths with general  (non-uniform) capacities. To the best of our knowledge, this is the first minmax regret result for dynamic flow problems in any type of graph with general   capacities.

\end{abstract}

\maketitle
\input{introduction_v2}
\input{prelims}
\input{reduction_v2}

\input{real_algorithm_v2}

\input{revised_upper_envs_v1}
\input{main_utility}

\input{forked_algorithm_v9}

\input{conclusion}

\bibliography{main2}
\bibliographystyle{plainnat}

\end{document}

%% file: introduction_v2.tex
\section{Introduction}
\label{sec:introduction}

Dynamic flow networks were introduced by Ford and Fulkerson in \cite{Ford1958a} to model flow over time.
The network is a graph $G=(V, E)$.
Vertices  $v \in V$ have  initial {\em weight} $w_v$ which is the amount  of flow  starting at $v$ to be moved to the designated sinks.
Each edge $e\in E$ has both  a  {\em length $d(e)$} and a {\em capacity $c(e)$} associated with it.  $d(e)$ denotes the  time required to travel between the endpoints of the edge;  $c(e)$ is  the amount of flow that can enter $e$
 in unit time.   If all the $c(e)$s  have the same value, the graph is said to have \textit{uniform capacity}.
The general problem is to move all flow from its initial vertices to sinks, minimizing designated metrics such as maximum  transport time.
 
Dynamic flow problems differ dramatically from standard network flow ones because the introduction of capacities leads to {\em congestion} effects that arise when flow reaching an edge $e$ needs to wait before entering $e$.
 

A vast literature on dynamic flows exists; see e.g., \cite{aronson1989survey,Fleischer2007}. Dynamic flows can also model evacuation problems \cite{Higashikawa2014g}. In this setting, vertices can represent rooms of the building, edges represent hallways, sinks are locations  that are emergency exits and the goal is to design a routing plan that evacuates all the people in the shortest possible time.  In the simplest version, the sinks are known in advance.  In the {\em sink-location} version, the problem is to place sinks that minimize the evacuation time.

Evacuation  is best modelled by  \textit{confluent} flow, in which all the flow that passes through a particular vertex must merge and travel towards the same destination. In the example above, confluence corresponds to an exit sign in a room pointing ``this way out'', that all evacuees passing through the room  must follow.

Min-cost confluent flows are hard  to construct in both the static and dynamic cases \cite{Chen2007,Dressler2010b,Shepherd2015};
Even finding a constant factor  approximate solution in the 1-sink case is NP-Hard.

Research on exactly  solving the sink-location problem has  therefore been restricted to special simpler classes of graphs such as paths and trees. On paths, the problem can be solved in time $\min (O(n+k^2\log^2n),O(n\log n))$ with uniform capacities, and in time $\min (O(n\log n+k^2\log^4n),O(n\log^3 n))$ when edges have general capacities \cite{BGYKK2017}. The $1$-sink problem on trees can be solved in $O(n\log ^2 n)$ time\cite{Mamada2006}. \cite{Higashikawa2014g,Bhattacharya2015} decrease this to $O(n\log n)$ on trees with uniform capacities. For the $k$-sink problem on trees, \cite{chengolin2016} solves the problem in $O(nk^2\log^5n)$ time, and the same authors reduced the time to $O(\max(k,\log n)kn\log^4 n)$ in \cite{chengolin2018}. This result holds for general capacities; a $\log$-factor can be shaved off in the uniform capacity setting. 

Robust optimization \cite{Kouvelis1997} permits  introducing uncertainty into the input. 
One way of  modelling  this is for the input not to specify an exact value $w_i$ denoting the initial supply at vertex $i$ but instead to only specify a range $[w_i^-,w_i^+]$ within which $w_i$ is constrained to fall.
Any possible input satisfying all the vertex range constraints is a (legal)
 \textit{scenario}. In this setting, the goal is to choose a center (sink-location) that provides a reasonable evacuation time for all possible scenarios.
  More formally, the objective is to find a center $x$ that minimizes \textit{regret} over all possible scenarios, where regret is the maximum difference between the time required to evacuate the scenario to $x$ and the optimal evacuation time for the scenario.
  Such minmax regret settings have been studied for many combinatorial problems \cite{aissi2009min} including $k$-median \cite{averbakh2003improved,Brodal2008} and $k$-center \cite{Averbakh1997,Yu2008,conf/cocoon/BhattacharyaK12}.  As the regret problem generalizes the basic optimization version of the problem, exact regret algorithms tend to be restricted  to simple (non NP-hard)  graph settings, e.g., on  paths and trees. 

The $1$-sink minmax regret problem on a path with uniform capacities is solved by  \cite{ChengHKNSX13} in time $O(n\log^2n)$. This was reduced to $O(n\log n)$ by \cite{Wang2014b,Higashikawa2014a}, and further to $O(n)$ by \cite{Bhattacharya2015}. For $k=2$, \cite{Li2014} proposed an $O(n^3\log n)$ algorithm which was later  \cite{Bhattacharya2015} reduced  to $O(n\log ^4 n)$ and then $O(n\log ^3 n)$ \cite{arumugam2019optimal}.
 For general $k$, \cite{arumugam2019optimal} gave two algorithms, one running in $O(n k^2\log^{k+1} n)$ and the other in $O(n^3 \log n)$. 

The $1$-sink minmax regret problem on uniform capacity trees can be solved in $O(n \log n)$ time \cite{Higashikawa2014g,Bhattacharya2015}.
\cite{benkoczi2019minmax} gives a $O(n^2)$ algorithm for the 1-sink minmax regret problem on a uniform capacity-cycle.

All of the results quoted assume uniform capacity edges. This paper derives a $O(n^4 \log n)$  
 algorithm to calculate min-max regret for {\em general} capacities on a path. We believe this is the first polynomial time algorithm for min-max regret for the  general capacity problem in any graph topology. The second note following \Cref {thm:reduction3} provides some intuition  as to why  the general capacity  problem is harder than the uniform one.

\begin{theorem} 
\label{thm:nlog4}
The 1-sink minmax regret location problem with general capacities on paths can be solved in $O(n^4 \log n)$ time. 
\end{theorem}

The paper is organized as follows.  \Cref{sec:prelims} introduces the formal problem definition and some basic properties.
\Cref{sec:reduction}  is the theoretical heart of the paper;  in \Cref{thm:reduction2}  it derives the existence of  
 a restricted set of scenarios, the {\em two-varying scenarios},  that is guaranteed to include at least one worst-case scenario for every input.
\Cref{sec: real alg} then shows how minimizing regret over the  two-varying scenarios  implies \Cref{thm:nlog4}  if the minimum value of a  certain set of  special functions  can be evaluated quickly. 
\Cref {sec: UE,sec: Main Utility} describe how, given certain facts about upper envelopes on lines, those special functions can be evaluated quickly.
Finally,   \Cref{sec: Utility}   proves the facts about upper envelopes.
The paper concludes in \Cref{sec:conc} with a short description of possible improvements and extensions.

 {\em Note: Similar to the center problem, the sink-location problem has two versions;  in the discrete version all sinks (centers) must be placed on a vertex.  In the continuous version, sinks (centers) may be placed on edges as well.  The version treated in this paper is explictly the continuous one but, with straightforward modifiations, the main  results, including \Cref{thm:nlog4}, can be shown to hold in the discrete case as well.
 }

%% file: prelims.tex
\section{Preliminaries}
\label{sec:prelims}

\subsection{Dynamic Confluent Flows on Paths}

\begin{figure}[t]
\centerline{
\includegraphics[width = 4.5in]{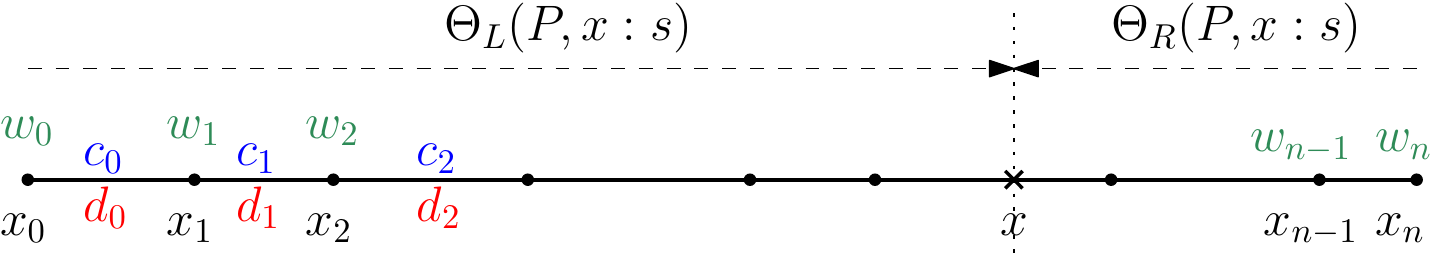}}
\caption{  Example of  the basic problem.  Point $v_i$ is situated at location $x_i$.  Edge $e_i=(x_i,x_{i+1})$. $c_i$ is the capacity of $e_i.$  $d_i = x_{i+1}-x_i$ is the length of $e_i.$ $w_i$ is the initial amount of (flow) items located on $v_i$. $\Theta_L(P,x:s)$ is the time required to evacuate all items to the left of  point $x$ to $x$.  $\Theta_R(P,x:s)$ is the time required to evacuate all items to the right of  point $x$ to $x$.  }
\label{fig: Fig1}  
\end{figure}

See \Cref{fig: Fig1}. The formal input to the {\em Dynamic Confluent Flow on a Path} problem is a path $P=(V,E)$ with 
$V = \{v_0,v_1,\ldots,v_n\}$ and $E= \{e_0,e_1,\ldots,e_{n-1}\}$ where $e_i = (v_i,v_{i+1})$.
\begin{itemize}
\item Each edge $e_i$ has associated {\em length} $d_i = d(v_i,v_{i+1}),$  which is the time required to travel between $v_i$ and $v_{i+1}.$ \\
$P$ is embedded on the line by placing vertex $v_i$ at location $x_i$  where
 $x_0=0$ and, $\forall i \ge 0$  $x_i = x_{i-1} + d_i = \sum_{j=0}^{i} d_j.$
\item $x \in P$ will denote any point, not necessarily a vertex, on the segment  $[0,x_n]$. 
\item Each edge $e_i$ also has an associated  {\em capacity}  $c_i=c(v_i,v_{i+1})$, denoting the amount of flow that can enter $e_i$ in unit time.  
\item  {\em Scenario} $s =(w_0(s),w_1(s),..,w_n(s))$.
$w_i(s)$ denotes the amount of flow initially starting at vertex $v_i$  in scenario $s.$ This flow needs to travel to a  {\em sink,}  $x \in P,$ where it will be evacuated. 
\item  For  $i \le j$, set $W_{i,j}(s)=\sum_{t=i}^j w_t(s)$. 

\end{itemize}

The basic problem is to find the  location of the sink $x \in P $ that minimizes the total evacuation time of all flow to $x$.  If the capacities are unbounded,  this reduces to the standard $1$-center problem. If capacities are bounded, \textit{congestion} can arise when too much  flow wants to enter an  edge.  This can occur in many different ways.
As an example, if flow is moving from left to right and  $c_{i-1}  > c_{i}$,  then flow enters $v_i$ from $e_{i-1}$ faster than it can leave $v_i$ to continue onto  $e_i.$    The congestion is caused by excess flow waiting  at $v_i$ until it can enter $v_i.$

Given a path $P=(V,E)$, lengths $d_i$, capacities $c_i$ and   scenario $s$, the time needed  to evacuate all the flow 
from the left of $x$, i.e., in $[0,x),$ to $x$ is denoted by $\Theta_L(P,x:s).$  The time needed  to evacuate all the flow 
from the right  of $x$, i.e., in $(x,x_n],$ to $x$ is denoted by $\Theta_R(P,x:s).$ The time needed to evacuate {\em all} of the flow to $x$ is the maximum of the left and right evacuation times: 
\begin{equation}
\label{eq:Theta Def}
\Theta(P,x:s)= \max\left \{\Theta_L(P,x:s), \Theta_R(P,x:s)\right \}.
\end{equation}

The formulae for the left and right evacuation times require a further definition:
%
\begin{definition} [Minimum capacity on a path]
\label{def:ctild}
Let $ x \le x'$ with $x,x' \in P$.
%
Set
$$c(x,x') = \min\{c_t \,:\, i \le t < j\}
\quad\mbox{where}\quad 
i = \max \{i'\,:\, x_{i'} \le x\}
\quad\mbox{and}\quad
j = \min \{j'\,:\, x_{j'} \ge x'\}.
$$
{\em Note: $c(x,x')$  is  the minimum-capacity of edges on the path connecting $x$ and $x'.$
}
\end{definition}

The formulae for $\Theta_L(P,x:s)$, and  $\Theta_R(P,x:s)$ are
\begin{lemma}\label{lem:evac} \cite{arumugam2019optimal}
\begin{equation}\label{eq:evac}
\Theta_L(P,x:s) = \max_{i:x_i < x} g_i(x :s)   
\quad\mbox{where}\quad
g_i(x :s) = 
\begin{cases}
d(x_i,x) + \frac{W_{0,i}(s)}{c(x_i,x)},  & \mbox{\rm if $W_{0,i}(s) > 0$},\\
0,  &   \mbox{\rm if $W_{0,i}(s) = 0$}.
\end{cases}
\end{equation}
\begin{equation}\label{eq:evaca}
\Theta_R(P,x:s) = \max_{i:x_i > x}h_i(x:s)  
\quad\mbox{where}\quad
h_i(x:s) =  
\begin{cases}
d(x,x_i) + \frac{W_{i,n}(s)}{c(x,x_i)}, & \mbox{\rm if $W_{i,n}(s) > 0$},\\
0,   & \mbox{\rm if $W_{i,n}(s) = 0$}.
\end{cases}
\end{equation}
\end{lemma}
For later use, we rewrite this for the different cases of $x$ being a vertex and $x$ on an edge:
\begin{corollary}
\label{cor:smooth evac}
\ 
\begin{enumerate}
\item 
If $x = x_j$ for some $j$, then
$$
\Theta_L(P,x:s) = \max_{i<j} g_i(x_j :s) 
\quad \mbox{and} \quad
\Theta_R(P,x:s) = \max_{i > j}h_i(x_j:s) .
$$
\item 
If $x \in ( x_j,x_{j+1})$ for some $j$, 
$$
\begin{array}{ccccccc}
\Theta_L(P,x:s)  &=& \max\limits_{i\le j} g_i(x_{j+1} :s)  - (x_{j+1}-x) & \quad \mbox{and}\quad& \Theta_R(P,x:s) &=& \max \limits_{i \ge  j+1} h_i(x_j,s)  -(x - x_j)\\[0.2in]
                            &=& \Theta_L(P,x_{j+1}:s)  - (x_{j+1}-x)                       &                                   &                            &=&\Theta_R(P,x_j:s) - (x - x_j).\\
\end{array}
$$
\end{enumerate}
\end{corollary}
We also need the following observations
\begin{corollary} Assume  $W_{0,n} > 0.$
\label{cor:Thetajump} 
\begin{itemize}
\item Let 
$t = \min\{ i  \,:\, W_{0,i} >0\}.$  
Then  $\Theta_L(P,x:s)=0$ for $x \le x_t$ and  $\Theta_L(P,x:s)$ is a monotonically increasing function of $x$ for $x \ge x_t.$
\item Let 
$t' = \max\{ i  \,:\, W_{i,n} >0\}.$  
Then  $\Theta_R(P,x:s)=0$ for $x \ge x_{t'}$ and  $\Theta_R(P,x:s)$ is a monotonically decreasing function of $x$ for $x \le x_{t'}.$

\item  $\Theta_L(P,x:s)$ and $\Theta_R(P,x:s)$ are piecewise linear and continuous everywhere except, possibly, at the vertices $x_j.$  

\item  $\Theta_L(P,x:s)$  is left continuous at $x_j$ and
 $\Theta_R(P,x:s)$  is right continuous at $x_j$, i.e.,
 $$\lim_{x \uparrow x_j} \Theta_L(P,x:s) = \Theta_L(P,x_j:s)
 \quad\mbox{and}\quad
 \lim_{x \downarrow x_j} \Theta_R(P,x:s) = \Theta_R(P,x_j:s),
 $$
but can have jumps at $x_j$ in the other directions.
\end{itemize}
\end{corollary}


\begin{definition}
The minimum time to evacuate for a given scenario over possible locations of sinks $x$ is denoted by 
$$\Theta_{OPT}(P:s) = \min_{x \in P} \left \{ \Theta(P,x:s) \right \}.$$
\end{definition}

The analysis will need the following basic concepts and observations:
\begin{definition}
\label{def:CRV}[Critical vertex]  
The left and right {\em critical vertices} of $x$ under scenario $s$ are
\begin{eqnarray*}
\lcrv(x:s) &=&  \argmax_{i:x_i < x} g_i(x :s),\\
\rcrv(x:s) &=& \argmax_{i:x_i > x} h_i(x :s).
\end{eqnarray*}
%
%
%
\end{definition}
$\lcrv(x:s)$, (resp.  $\rcrv(x:s)$)  is  the vertex  at which the maximum value that defines the left (resp. right)  evacuation cost  is achieved. In the case of ties in achieving the maximum,   $\argmax$ will choose $\lcrv(x:s)$ (resp. $\rcrv(x:s)$)  to be the maximizing index $i$  closest to $x.$

From \Cref{cor:Thetajump},
$\Theta_L(P,x:s)$ and $\Theta_R(P,x:s)$ are, respectively, monotonically increasing and decreasing nonnegative functions in $[x_0,x_n]$ (except, respectively, for an interval at  their start and finish where they might be  identically zero).  Thus, from \Cref{eq:Theta Def}, 
$\Theta(P,x:s)$ first monotonically decreases in $x$, reaches a minimum and then monotonically increases in $x.$

\begin{definition}[Unimodality]
Let $f(x)$ be a function defined over interval $I = [\ell,r] \subseteq \Re$.
$f(x)$ is {\em Unimodal} 
over interval $I$,  if\  $\exists x^* \in I$ such that
$f(x)$ is monotonically decreasing in $[\ell,x^*)$ and  monotonically increasing in $(x^*,r]$.  $x^*$ is called the {\em mode} of $f$.
\end{definition}

The discussion  preceding the definition and the fact  that $\Theta_L(P,x:s)$ and $\Theta_R(P,x:s)$ are continuous everywhere except, possibly, at  the points $x_i$, implies the following:
%
%
\begin{observation}
\label{obs: unimodal}
$\Theta(P,x:s)$ is a unimodal function in $x$ over $P$.  Furthermore, it is continuous everywhere except, possibly, at the points $x_i.$
\end{observation}

Finally, we define
\begin{definition}[Optimal Sink]  The {\em Optimal Sink} for scenario $s $ is $x_{\OPT}(s) \in [x_0,x_n]$ such that
$$\Theta(P,x_{\OPT}(s):s) = \Theta_{\OPT}(P:s).$$
This optimal sink is unique because $\Theta(P,x:s)$ is a unimodal function of $x.$
\end{definition}

By standard binary searching techniques,
\begin{observation}
\label{obs:bs unimodal}
Let $I = [\ell,r] \subseteq [x_0,x_n]$, $f(x)$ a unimodal  function over $I$ and $x^*$ the mode of  $f(x)$ in $U.$
 The evaluation of $f(x_i)$ for some $x_i \in I$ will be denoted as a ``query''.  

Then the unique index $i$ such that $x^* \in [x_i,x_{i+1})$  or the fact that  $x^* = x_n$ can be found using $O(\log n)$ queries. 
\end{observation}

For 
later use, the following will also be needed.
\begin{lemma}
\label{lem:max unimodal}
Let $\cal I$  be some index set (finite or infinite) and $\{f_z(x) :z \in {\cal I}\}$ be  a set of functions, all unimodal in an interval $I$.  Then, if 
$$f_{\max}(x) = \max_{z\in Z} f_z(x)$$
exists, it is also a unimodal function in $I.$
\end{lemma}
\begin{proof}
First suppose that $f$ and $g$ are both unimodal functions with $x^*_f$,   $x^{*}_g$ being, respectively, the  unique minimum locations of  of $f$ and $g$. Set $h=\max(f(x),g(x)).$

Without loss of generality assume 
 $x^*_f \le x^{*}_g$.   Then $h(x)$ is monotonically decreasing for $x <  x^*_f$ and monotonically increasing for $x \ge x^{*}_g.$
Now consider  $x\in I'=[x^*_f,x^{*}_g]$: $f(x)$ is monotonically increasing in $I'$, while $g(x)$ is monotonically decreasing in $I'$.

If $f(x^*_f) \ge g(x^*_f) $ then  
$$\forall x\in I',\quad  f(x)  \ge  f(x^*_f)  \ge g(x^*_f)  \ge  g(x)$$
so  $h$ is unimodal with mode $x^*_f$.  
If $f(x^*_g) \le  g(x^*_g) $ then 
$$\forall x\in I',\quad  f(x)  \le  f(x^*_g)  \le g(x^*_g)  \le  g(x)$$
so  $h$ is unimodal with mode $x^*_g$.  

Otherwise $f(x^*_f) <  g(x^*_f) $ and  $f(x^*_g) >  g(x^*_g) $ so 
 $\bar x  = \sup\{x \in I' \,:\, f(x) \le g(x)\}$ exists and 
 $$\forall x \in [x^*_f,\bar x],\,  f(x) \le g(x) \mbox{ so } h(x) = g(x)
 \quad\mbox{and}\quad 
 \forall x \in (\bar x, x^*_g],\,  f(x)  > g(x) \mbox{ so } h(x) = f(x)
 $$
 and 
 $h$ is unimodal with mode $\bar x$.

Repeating this  process yields that for any three unimodal functions $f(x), g(x), u(x)$,  the function $\max\left( f(x),g(x),u(x)\right)$ is also unimodal.

Now suppose that 
$f_{\max}(x)$ is not unimodal.  Then there exist  3 points $x_1 < x_2 < x_3$ such that
$f_{\max}(x_1) <  f_{\max}(x_2)  >  f_{\max}(x_3)$.  Then there exists three functions $f_1,f_2,f_3$  in the set $\{f_z \,:\, z \in Z\}$ satisfying
$f_{\max}(x_i) = f(x_i)$ for $i=1,2,3.$  But this contradicts the fact that 
$\max\left( f_1(x),f_2(x),f_3(x)\right)$ is also unimodal.
\end{proof}

\subsection{Regret}
One method for capturing uncertainty in the input is the 
{\em min-max regret} viewpoint. In this,  the vertex weights are not fully specified in advance. Instead, a range of weights  $[w_i^-,w_i^+]$ in which 
$w_i(s)$ must lie is specified. A specific assignment of weights to the vertices is a  {\em legal} scenario. The set of all possible legal  scenarios is the Cartesian product of all possible ranges for the weights. 
\[
\calS =  \prod_{i=0}^n [w_i^-,w_i^+]
\]

\begin{definition}
\label{def:sstuff}
 Let   $i,j$ satisfy $0 \le i \le j \le n$ and   $\alpha, \beta \ge 0$. 
\begin{itemize}
\item
Let $s$ be  a scenario. Set $s_{-i,-j}(\alpha, \beta)$ and $s_{-i}(\alpha)$ to be the  unique scenarios satisfying
$$
w_t\left( s_{-i,-j}(\alpha, \beta) \right)=
\begin{cases}
 w_t(s)   &\mbox{if  $t \not\in \{i,j\}$},\\
  \alpha &  \mbox{if \ $t =i$},\\
\beta &  \mbox{if \ $t =j$}.\\
\end{cases},
\hspace*{.2in}
w_t\left( s_{-i}(\alpha) \right)=
\begin{cases}
 w_t(s)   &\mbox{if  $t \not=i$},\\
  \alpha &  \mbox{if \ $t =i$}.
\end{cases}
$$
Note that in some proofs, we will have $i=j.$  In those cases, it will be required that $\alpha = \beta$  so that $s_{-i,-j}(\alpha, \beta) = s_{-i}(\alpha)$.
\item  Set $s'_{i,j}$ to be the scenario satisfying
$$
w_t\left( s'_{i,j}\right)=
\begin{cases}
 w_t^- & \mbox{if \ $t <i$ or $j < t$},\\
  w_t^+   & \mbox{if \ $i< t < j$},\\
 0&  \mbox{if \ $t \in \{i,j\}$}.\\
\end{cases}
$$
\item See  \Cref {fig: Fig2}. Define  $s_{i,j}(\alpha,\beta) =  s_{-i,-j}(\alpha, \beta)$ where $ s = s'_{i,j}.$
\item For any fixed $i,j,$  the corresponding set of {\em two-varying scenarios} is
$$ S_{i,j} 
= \left\{ s_{i,j} (\alpha,\beta) \,:\,  
 \alpha \in \left[w_i^-, w^+_i\right] \mbox{   and  }
\beta\in \left[w_j^-, w^+_j\right]
 \right\}.$$
\end{itemize}
\end{definition}

\begin{figure}[t] 
\centerline{\includegraphics[width = 6in]{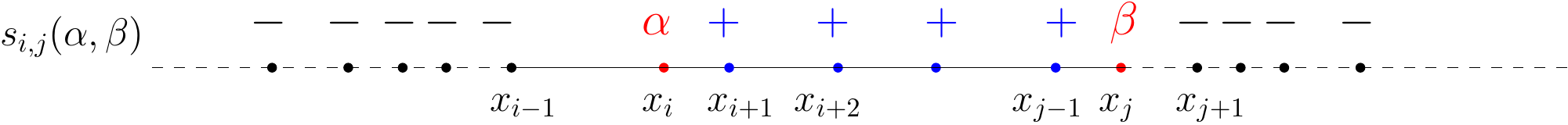}}
\caption{  Illustration of two varying scenario  $s=s_{i,j}(\alpha,\beta).$   All vertices $v_t$ with $t < i$ or $t > j$ have  $w_t(s) = w_t^-$ and 
all  vertices $t$ with $i < t < j$ have  $w_t(s)=w^+_t.$  
Finally $w_i(s) = \alpha$ and $w_j(s) = \beta$.}
\label{fig: Fig2}  
\end{figure}

In all the definitions that follow,  input path $P=(V,E)$, is considered as fixed and given.

\begin{definition}[Regret for $x$ under scenario $s$]\label{def:regret}
The {\em regret for a location $x$ under scenario $s$} is
\begin{equation}\label{eq:def-regret}
R(P,x:s) = \Theta (P,x:s)-\Theta_{OPT}(P:s).
\end{equation}
This  is  the difference between evacuation time to  $x$ and the optimal evacuation time.
\end{definition}

\begin{definition}[Max-regret for $x$]\label{def:maxregret}
The {\em Max-regret} over all possible legal scenarios for $x$ is:
\begin{equation}\label{eq:def-maxregret}
R_{\max}(P,x)=\max_{s\in \mathcal{S}} \left \{ R(P,x:s) \right \}
\end{equation}
\end{definition}

\begin{definition}[Worst-case scenario]\label{def:worstsce}
Scenario $s\in \calS$ is a  {worst-case scenario} for $x$ if 
\begin{equation}
R_{\max}(P,x)=R(P,x:s)
\end{equation}
\end{definition}

\begin{definition}[Minmax regret]\label{def:minmaxregret}
The {\em Min-Max Regret} value is 
the minimum possible Max-regret over all possible locations $x$:
\begin{equation}\label{eq:minmaxregret}
R_{OPT}(P)=\min_{x \in P}\left \{R_{\max}(P,x)\right \}
\end{equation}
\end{definition}

The goal is to find $x\in P$ such that $ R_{OPT}(P)= R_{\max}(P,x).$
\subsection{Technical Observations  for later use}

\begin{lemma}
\label{lem:rmax unimodal}
Let $I =[x_0,x_n]$. Then
\begin{enumerate}
\item For every  $s \in \calS$, $R(P,x:s)$ is  a unimodal function of $x$ over $I$.
\item $R_{\max}(P,x)$ is  a unimodal function of $x$ over $I$.
\end{enumerate}
Furthermore $R(P,x:s)$  and $R_{\max}(P,x)$ are continuous everywhere  except, possibly, at the  $x_i.$
\end{lemma}

\begin{proof}
(1) follows from Observation \ref{obs: unimodal} and the fact that subtracting  a constant from a unimodal function leaves a unimodal function.

(2) follows from (1) and \Cref{lem:max unimodal}.

The continuity of $R(P,x:s)$  follows directly from the continuity of  $\Theta (P,x:s)$; the continuity of $R_{\max}(P,x)$ follows from the continuity of $R(P,x:s)$.
\end{proof}

The following technical lemma will be needed later for the correctness of the algorithm.

%
%

\begin{lemma}
\label{lem:sidelimits}
Let $ \ell < r$ and  $f_\ell$, $f_r$,  $g_\ell$ and $g_r$ be constants satisfying
\begin{equation}
\label{eq:AA1}
f_\ell \le f_r- (r- \ell),\quad
g_r \le g_{\ell}-(r-\ell).
\end{equation}
For all $x \in [\ell,r]$, define
$$
f(x) = 
\begin{cases}
f_\ell  & \mbox{if $x = \ell$},\\
f_r- (r -x) & \mbox{if $x \in (\ell,r)$, }\\
f_r& \mbox{if $x = r$},\\
\end{cases}
\quad\mbox{and}\quad
g(x) = 
\begin{cases}
g_\ell   & \mbox{if $x = \ell$},\\
g_\ell - (x - \ell)& \mbox{if $x \in (\ell,r)$, }\\
g_r& \mbox{if $x =r$}
\end{cases}
$$
and 
$$
\bar f(x) = 
\begin{cases}
f_r - (r - \ell)   & \mbox{if $x = \ell$},\\
f(x) & \mbox{if $x \in (\ell,r]$, }\\
\end{cases}
\quad\mbox{and}\quad
\bar g(x) = 
\begin{cases}
g(x)   & \mbox{if $x  \in [\ell,r)$},\\
g_r& \mbox{if $x =r$.}
\end{cases}
$$
Note that $f(x)$ (resp. $g(x)$) is continuous everywhere except possiby at  the endpoint $x= \ell$ (resp. $x=r$) while $\bar f(x)$ (reps. $\bar g(x)$)) is continuous everywhere.

Finally for  all $x \in [\ell,r]$ define
$$
 h(x) = \max\left( f(x), g(x)\right)
\quad\mbox{and}\quad
\bar h(x) =  \max\left(\bar f(x), \bar g(x)\right)
$$
Then 
$\min_{x \in  [\ell,r]} h(x)$ exists and   
$$
\min_{x \in  [\ell,r]} h(x) =
 X
\quad\mbox{where}\quad 
X = 
\min\left(
h(\ell),\, h(r),\,  \min_{x \in  [\ell,r]} \bar h(x)
\right).
$$
\end{lemma}

\begin{proof}
Technically, because $h(x)$ might not be continuous at $\ell,r,$ only the existence of   $\inf_{x \in  [\ell,r]} h(x)$  is known.  
Proving the lemma  requires also proving that $\min_{x \in  [\ell,r]} h(x)$  also exists.

From the definitions,
$$h(\ell) = \max\left(f(\ell),g(\ell)\right) \le \max\left(\bar f(\ell),\bar g(\ell)\right) = \bar h(\ell)$$
and 
$$h(r) = \max\left(f(r),g(r)\right) \le \max\left(\bar f(r),\bar g(r)\right) = \bar h(r).$$

Since $\bar h$ is a continuous bounded function in  $[\ell,r]$, $Y=\min_{x \in  [\ell,r]} \bar h(x)$ exists, so $X$  also exists.

Now let  $x^* \in [\ell,r]$  be such that $Y=\bar h(x^*)$.  If $x^* \in (\ell,r)$, then since $h(x) = \bar h(x)$ for all $x \in (\ell,r)$,
$$
\inf_{x \in  [\ell,r]} h(x)=
\min\left(  h(\ell),\, h(r),\,  h(x^*)\right) = X.
$$

If $x^* \not\in (\ell,r)$ then $x^* \in \{\ell,r\}$   so  $Y = \min\left(\bar h(\ell),\bar h(r)\right)$.  Since  $h(\ell)  \le \bar h(\ell)$ and  $h(r)  \le \bar h(r)$,
$$
\inf_{x \in  [\ell,r]} h(x) =
\min\left(  h(\ell),\, h(r),\,  \bar h(\ell),\,  \bar h(r)\right) 
= \min\left(  h(\ell),\, h(r)\right) =X.
$$

In both the cases, we have shown that $\inf_{x \in  [\ell,r]} h(x)=X$ and that the infimum  is achieved at some value $h(x)$, $x \in [\ell,r]$ so 
$\min_{x \in  [\ell,r]} h(x) = \inf_{x \in  [\ell,r]} h(x) = X.$

\end{proof}

The later proof of \Cref {lem: fork H cor} will also require the following  simple corollary:
\begin{corollary}
\label{cor:techlimit}
Fix $k$ and some scenario $s \in {\calS}.$  Then
$$
\min_{y \in [x_k,x_{k+1}]}  \Theta(P,y:s) =
\small 
\min
\small
 \left(
\Theta(P,x_k:s),\, 
 \Theta(P,x_{k+1}:s),\, Y
  \right).
  $$
  where
  $$ Y = 
\min_{y \in [x_k,x_{k+1}]} 
\max
\bigl\{
 \Theta_L(P,x_{k+1}:s) - (x_{k+1}-x), \Theta_R(P,x_{k}:s) - (x -x_k)
\bigr\}
$$
\end{corollary}
\begin{proof}
Set   $ \ell = x_k$,  $r = x_{k+1}$ and 
$$
\begin{array}{ccccccc}
f_\ell  &=&   \Theta_L(P,x_{k}:s), & \quad & f_r  &=&   \Theta_L(P,x_{k+1}:s),\\
g_\ell  &=&   \Theta_R(P,x_{k}:s), & \quad & g_r  &=&   \Theta_R(P,x_{k+1}:s).
\end{array}
$$
From  \Cref{cor:smooth evac,cor:Thetajump},  \Cref{eq:AA1} is satisfied.
Also from  \Cref{cor:smooth evac},  for $x \in (x_k,x_{k+1})$
$$
\begin{array}{cllcc}
f(x) &= &\Theta_L(P,x_{k+1}:s) - (x_{k+1} - x) &=& \Theta_L(P,x:s),\\
g(x) &= &\Theta_R(P,x_{k}:s) - (x- x_k) &=& \Theta_R(P,x:s).
\end{array}
$$
The proof of the Corollary then follows directly from~\Cref{lem:sidelimits}.
\end{proof}

%% file: reduction_v2.tex
\section{Reduction to scenarios with two varying weights}
\label{sec:reduction}
As the scenario space  $\calS$ is  infinite,
it is impossible to  calculate  $R_{\max}(P,x)$ directly from 
\Cref{eq:def-maxregret}.  
 To sidestep  this, the standard approach, e.g.,  \cite{ChengHKNSX13,Wang2014b,Higashikawa2014a,Bhattacharya2015},   for the uniform capacity case has been to first reduce the scenario space to a finite set of  possible {\em worst-case} scenarios.

 As the first step in this direction  for general capacities,  we start with a lemma describing  the effect of changing weights of a single vertex in a scenario. 

\begin{definition} Let $s \in \calS.$  $s' \in \calS$ is {\em  obtained from $s$} by the operation $\shift (i,j,\delta)$  by setting 
$$
w_t(s')=
\left\{
\begin{array}{ll}
w_i(s)-\delta   &\mbox{if $t=i$},\\
w_j(s)+\delta &\mbox{if $t=j$},\\
w_t(s) & \mbox{if $t\not=i,j$}.\\
\end{array}
\right.
$$
Note that $\shift (i,j,\delta)$  is a valid operation only if
$w_i(s)  \ge w_i^-+\delta$ and $w_j \le w_j^+-\delta$. 
\end{definition}

\begin{lemma}\label{lem:shift}
Let $s'\in\calS$ be  obtained from $s \in \calS$ by applying a valid  $\shift (i,j,\delta)$ operation.
\begin{itemize}
\item[(a)] If  $i<j$ and $ x_j\le x$, then  $\Theta(P,x:s')\leq \Theta(P,x:s).$ 
\item[(b)]  If  $j < i$ and  $x \le x_i$, then  $\Theta(P,x:s')\leq \Theta(P,x:s).$ 
\end{itemize}
\end{lemma}

\begin{proof} We prove (a).  The proof of (b) is symmetric.

Consider the formula in~\Cref{lem:evac}. For every $k$, $W_{0,k}(s') \le W_{0,k}(s)$. Thus, $\Theta_L(P,x:s') \leq \Theta_L(P,x:s)$. Furthermore, for every $k$ satisfying  $x \le x_k$,
$W_{k,n}(s') =W_{k,n}(s)$.  Thus 
$\Theta_R(P,x:s')=\Theta_R(P,x:s)$. The proof of the Lemma follows immediately.
\end{proof}

Given a set of scenarios ${\cal S}' \subseteq {\cal S}$,  vertex $x_i$ is   \textit{varying} in ${\cal S}'$ if  $w_i(s)$ not constant  for all $s \in {\cal S}'.$
In the full set $\cal S$ of all possible scenarios, all the vertices  might  be varying. 
The important observation  will be  that, when considering sets of worst-case scenarios, it suffices to consider subsets in which only two vertices are varying and the rest have fixed weights, i.e.,~the sets $S_{i,j}$.

\begin{figure}[t] 
\centerline{\includegraphics[width = 6in]{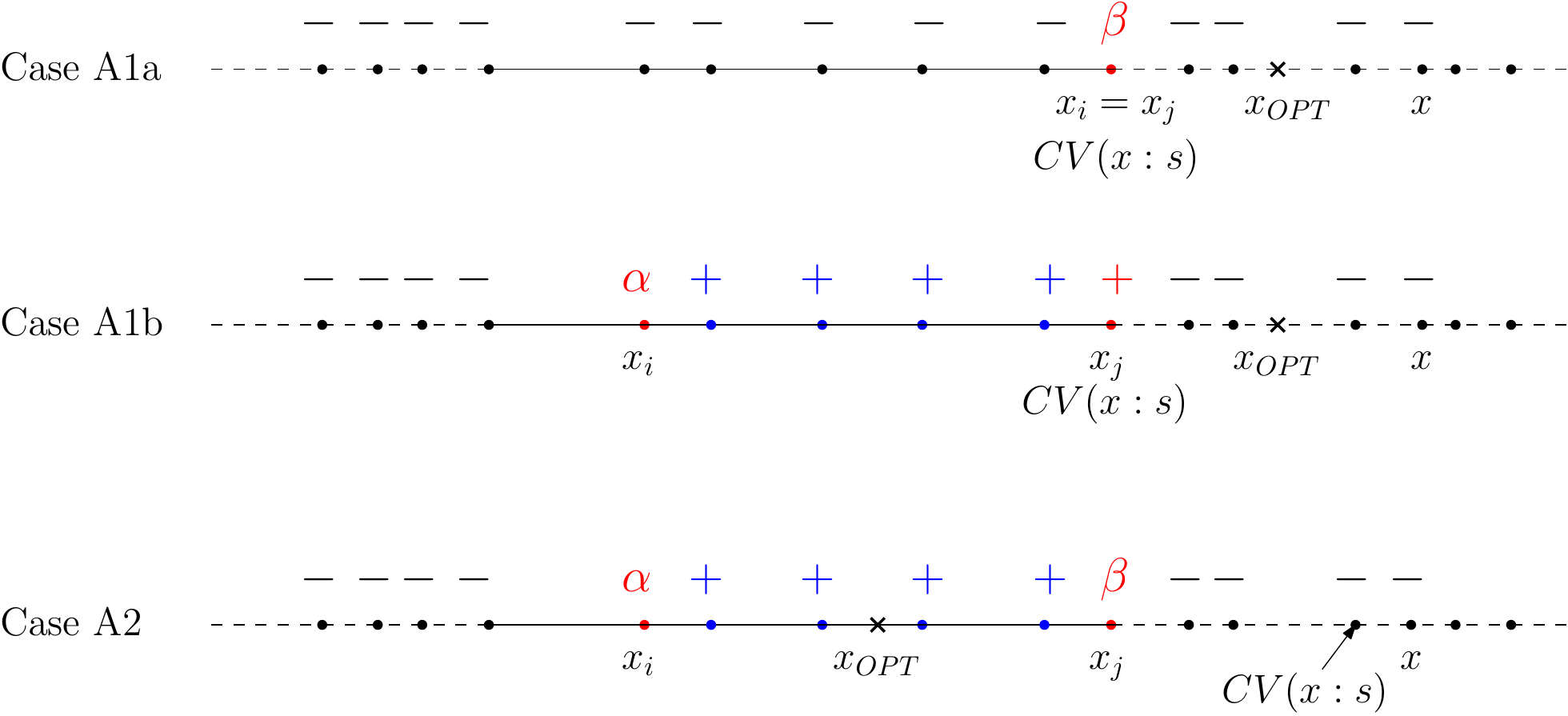}}
\caption{  Illustration of the three possible worse-case scenarios that can occur in Case A in \Cref{thm:reduction2}. In all three cases, $\Theta(P,x)$ is the time for evacuation from the left.
$CV(x:s)$ is the index of the critical vertex in the evacuation from the left. 
 $x_{\OPT}$ is the location of the optimal sink (that minimizes evacuation time).
In case A1a, $i=j$ while  in case A1b, $i < j$ but $w_j(s) = w^+_j$.  In both A1 scenarios,  $x_j$ is the critical vertex and  
$x_j \le x_{\OPT} \le x$. 
In case 2, the critical vertex can be anywhere between $x_j$ and $x$  and  $x_i \le x_{\OPT} \le x_j.$
}
\label{fig: Fig3}  
\end{figure}

\begin{theorem}\label{thm:reduction2} Let $x \in P.$  There exists  a worst case scenario $s$  for $x$, such that $s \in S_{i,j}$ for some $i <j$  and at least one of the following conditions is true.
See \Cref {fig: Fig3}.
\begin{itemize}
\item[(A)]  $\Theta(P,x:s)=\Theta_L(P,x:s)$ and 
	\begin{enumerate}
	  \item 
	           $\lcrv(x:s) = j$  where   $x_j  \le x_{\OPT}(s) \le x$,  $x_j < x$  and 
                \begin{itemize}
		\item[(a)] Either  $i=j$  (so $s\in S_{j,j}$)
		\item[(b)] or  $i<j$  and $w_j(s) = w_j^+$
                \end{itemize}
	  \item or  $\lcrv(x:s) \ge j$ and $x_i \le x_{ \OPT}(s) \le x_j \le x$
	\end{enumerate}
\item[(B)]  $\Theta(P,x:s)=\Theta_R(P,x:s)$ and
	\begin{enumerate}
	  \item    $\rcrv(x:s) =i$ where      $  x \le x_{\OPT}(s)  \le x_i$,  $x < x_i$   and 
                \begin{itemize}
		\item[(a)] Either  $i=j$  (so $s\in S_{i,i}$)
		\item[(b)] or  $i<j$  and $w_i(s) = w_i^+$
                \end{itemize}
	  \item  or $\rcrv(x:s) \le i$ and $x \le x_i \le x_{\OPT}(s) \le x_j$
	\end{enumerate}
\end{itemize}

\end{theorem}

\begin{proof}
Let $s$ be a worst-case scenario for $x$. We apply a series of transformations to $s$, converting it to $s^*$, maintaining the property that  at each step, the currently constructed  
scenario remains a worst-case scenario for $x$. The final constructed scenario $s^*$  will  satisfy the conditions of the theorem.

Without loss of generality, assume that $\Theta(P,x:s)=\Theta_L(P,x:s)\ge\Theta_R(P,x:s)$. We prove  that a worst-case $s'$  satisfying (A) exists. The proof for (B) when
$\Theta(P,x:s)=\Theta_R(P,x:s)$ will be totally symmetric.

First note that  if $ x < y$, then by the monotonicity of both $\Theta_L(P,y:s)$ and $\Theta_R(P,y:s)$ in $y$,
$$\Theta_L(P,x:s)  <  \Theta_L(P, y:s) 
\quad\mbox{and}\quad
\Theta_R(P,x:s)  >  \Theta_R(P, y:s).
$$
Thus
$$\Theta(P, y:s)  = \max\{ \Theta_L(P, y:s), \Theta_R(P, y:s) \} > \Theta_L(P,x:s) = \Theta(P, x:s)
$$
and $y \not= x_{\OPT}(s).$ Thus, $x_{\OPT}(s) \le x.$

For notational  simplicity, now  set  $k=\lcrv(x:s)$ 
and  $y  = x_{\OPT}(s)$. By the argument above,  $y \le x$.  Also, by the definition of $\lcrv(x:s)$, $x_k < x.$

Next, change the weights of all vertices $i>k$ to $w_i^-$. Denote the resulting scenario as $s'$.  
As no weights are increased, $\Theta_R(P,x:s') \le \Theta_R(P,x:s)$ and  $\Theta_{OPT}(P:s') \leq \Theta_{OPT}(P:s)$.
Note that the weights in  $[v_0,v_k]$ are unchanged and weights to the right of $v_k$ can  only decrease.  From the definition of  $k=\lcrv(x:s)$,  
$$\Theta_L(P,x:s')=
g_k(x:s') = g_k(x:s) = \Theta_L(P,x:s) \ge\Theta_R(P,x:s) \ge\Theta_R(P,x:s').$$
Thus,
\begin{eqnarray*}
 \Theta(P,x:s') &=& \max\left \{\Theta_L(P,x:s'), \Theta_R(P,x:s')\right \}\\
 			&=& \Theta_L(P,x:s)\\
 			&=&\Theta(P,x:s).
\end{eqnarray*}
This yields
$$ \Theta (P,x:s')-\Theta_{OPT}(P:s') \ge \Theta (P,x:s)-\Theta_{OPT}(P:s)$$
Thus,  $s'$ is also a worst case scenario for $x$. 
Henceforth, we may assume that $w_i(s)=w_i^-$ for $i>k$.

\paragraph{Case-1: $y = x_{\OPT}(s) \ge  x_k$} \ \\
As long as there exists a $i,j,\delta$ that permits a valid shift operation,
 apply $\shift(i,j,\delta), i<j \leq k$. Utilize the rule of always choosing  the smallest $i$ that currently has more than $w_i^-$ flow and largest $j$ that currently has less than  $w_j^+$ flow. 
 When applying  the shift  to $i,j$,  choose the largest possible $\delta$ that maintains validity. Then, either $w_i$ is set to  $w^-_i$ or $w_j$ is set to $w^+_j$. This process must therefore conclude after at most $k$ shift operations. Denote the final resulting scenario by $s'$. 
 This construction is pushing units of flow to the right, towards $k$. Thus by construction,  it must end in one of the following two cases
 \begin{itemize}
 \item If $w_k(s') < w_k^+$, then $\forall t \le k,$  $w_t(s') = w^-_t$, so $s' \in S_{k,k}.$ 
 \item If $w_k(s') = w_k^+$, let $i < k$ be the largest index such that $w_i(s') > w_i^-.$ Then $\forall t < i$, $w_t(s') = w_t^-$  and $\forall i < t \le k$,  $w_t(s') = w^+_t$ so $s' \in S_{i,k}$,
 \end{itemize}
 The constructed scenarios are then of type A1a or A1b after setting $j=k$. Note that this implies $x_j = x_k < x.$
 

To prove that this set of operations indeed preserves the worst-case property, first note that, by construction,
 $W_{0,t}(s') = W_{0,t}(s)$ for all $t \ge j.$  In particular,  this implies that
  $W_{0,j}(s') = W_{0,j}(s)$.  Since
 $j=k=\lcrv(x:s),$  from the definition of $g_j(x,s),$
\begin{equation}\label{eq:change1}
\Theta(P,x:s') \ge g_j(x:s') = g_j(x,s)  =  \Theta(P,x:s).
\end{equation}
 Set $y' = x_{\OPT}(s')$ and  $j'=\lcrv(x:s')$. (A-priori, it is possible that $j'\not=j$ or  $y \not = y'.$)
 
 From \Cref{lem:shift}(a)  and the assumption $y \ge x_k = x_j,$ 
\begin{equation}\label{eq:change2}
\Theta(P, y' :s')=  \Theta_{OPT}(P:s') \leq \Theta(P,y :s') \leq \Theta(P,y :s) = \Theta_{OPT}(P:s). 
\end{equation}
Thus, from~\Cref{eq:change1,eq:change2}, 
\begin{eqnarray*}
\Theta(P,x:s')   -  \Theta_{OPT}(P:s')  &\ge &\Theta(P,x:s')  - \Theta_{OPT}(P:s) \\
						        &\ge& \Theta(P,x:s)  - \Theta_{OPT}(P:s) =R_{\max}(P,x)
\end{eqnarray*}
so  $s'$ is also a worse-case scenario for $x$.

To complete the proof of Case 1 for $s'$,  it suffices to show that $ j=j'$ and $y=y'$. Note that $\Theta(P,x:s') \geq \Theta(P,x,s)$ and $\Theta_{OPT}(P:s') \leq \Theta_{OPT}(P:s) $. If either  $ \Theta(P,x:s')  > \Theta(P,x:s),$  or $\Theta_{OPT}(P:s')< \Theta_{OPT}(P:s) $ then
$$\Theta(P,x:s')   -  \Theta_{OPT}(P:s') > R_{\max}(P,x),$$
contradicting the definition of $ R_{\max}(P,x)$.  Thus $\Theta_{OPT}(P:s')=  \Theta_{OPT}(P:s) $  and  $ \Theta(P,x:s')  = \Theta(P,x:s).$ 
Plugging  the first equality into \Cref{eq:change2} immediately yields
\begin{equation}
\label{eq:A}
\Theta(P,y':s')=\Theta_{OPT}(P:s') = \Theta(P,y:s') = \Theta_{OPT}(P:s)
\end{equation}
so
$$y'=y.$$
Furthermore, the  second equality and the fact that 
 $W_{0,t}(s') = W_{0,t}(s)$ for all $t \ge j$ immediately implies that 
 $$j'=\lcrv(x:s') = \lcrv(x:s) =j.$$

\paragraph{Case-2: $y = x_{\OPT}(s) <x_k$} \ \\
Let  $u<k$ be  the largest index such that $x_u \leq y <x_k$.

First apply $\shift(i,j,\delta), i<j\leq u$ if there exists a $i,j,\delta$ that permits a valid operation, using the same rule as in Case-1.
 By the same argument as in Case 1, these shifts can be shown to  terminate after a finite number of steps.   Let $s'$ be the resulting scenario.
 Note that after completing these shifts, one of the following two situations must occur:
\begin{itemize}
\item Either $w_u(s')  < w^+_u$,  and  thus $w_t(s') = w^-_t$ for all $t < u$.
\item or $w_u(s')  = w^+_u$ and  there exists $p < u$ such that $w_t(s') = w^-_t$ for all $t < p$ and
$w_t(s') = w^+_t$ for all $p  < t \le u.$
\end{itemize}

 Next, apply $\shift(j,i,\delta), y\le x_i <x_j \leq x_k$ if there exists a $i,j,\delta$ that permits a valid operation. Now use the rule of always choosing  the smallest $i$ that currently has less than $w_i^+$ flow and largest $j$ that currently has more than  $w_j^-$ flow. Again, always choose  the largest possible $\delta$ valid for the $i,j$ pair.
 These operations also terminate after a finite number of steps.
After the termination of this seconds set of  shifts, the new resulting scenario $s'$ satisfies
\begin{itemize}
\item There exists $q,$  $u  \le  q \le k$  such that
$w_t(s') = w^+_t$ for all  $u < t < q $   and  $w_t(s') = w^-_t$ for all  $q < t \le k$  
\end{itemize}
Combining this with the fact that, 
\begin{itemize}
\item $\forall t \ge k,$   $w_t(s') = w_t(s) = w^-_t$
\end{itemize}
yields that  $s' \in S_{p,q}$ with $x_p \le  x_u \le y  \le x_q \le x_k.$

 Set $y' = x_{\OPT}(s')$ and  $k'=\lcrv(x:s')$.
 
The proof that $s'$ is also a  worst-case scenario for $x$ is similar to that in Case 1.

By construction,
 $W_{0,t}(s') = W_{0,t}(s)$ for all $t \ge k.$  In particular,  this implies that
  $W_{0,k}(s') = W_{0,k}(s)$.  Since
 $k=\lcrv(x:s)$,
 $$ \Theta(P,x:s') \ge g_k(x:s') = g_k(x,s)  =  \Theta(P,x:s).$$

Both types of shifts (to the right and to the left) applied above 
 satisfy the conditions in ~\Cref{lem:shift}((a) and (b)) with respect to $y=x_{OPT}$. Applying ~\Cref{lem:shift}  and applying the same argument following ~\Cref{eq:change1}   yields that
 ~\Cref{eq:A} is correct for this case as well.

The exact same arguments as in Case-1  then show that 
$s'$ is a worst-case scenario for $x$,  $y'=y$ and $k'=k$ completing the proof of Case 2 after setting $i =p$ and $j=q$. Note that $x_j = x_q \le x_k < x.$
%
\end{proof}


\begin{definition} 
\label{def:gandh defs}
 Let $ i < j $. Set 
$$
G_{j} (x) =\max_{ s \in S_{j,j}} \left\{ g_j(x:s) - \min_\yjyx\Theta(P,y:s)  \right\}
\quad\mbox{and}\quad
H_{i} (x) =   \max_{ s \in S_{i,i}} \left\{ h_i(x:s) - \min_\yxyi \Theta(P,y:s) \right\}.
$$
These correspond, respectively,  to Cases A1a  and B1a of \Cref{thm:reduction2}.

Now set 
\begin{align*}
G_{i,j} (x) &=&& \max_{ s \in S_{i,j}, w_j(s) = w_j^+} \left\{ g_j(x:s) - \min_\yjyx \Theta(P,y:s) \right\},\\
\bar G_{i,j}(x) &=&& \max_{ s \in S_{i,j}} \left\{ \left(\max_\tjtx g_t(x:s)\right) - \min_\yiyj \Theta(P,y:s) \right\},\\
H_{i,j} (x) &=&&  \max_{ s \in S_{i,j}, w_i(s) = w_i^+} \left\{ h_i(x:s) - \min_\yxyi \Theta(P,y:s) \right\},\\
\bar H_{i,j}(x) &=&& \max_{ s \in S_{i,j}} \left\{ \left(\max_\txti h_t(x:s)\right) - \min_\yiyj \Theta(P,y:s) \right\},\\
\end{align*}
These correspond, respectively,   to Cases A1b, A2, B1b and B2 of  
\Cref{thm:reduction2}.
\end{definition}

Next, define

\begin{definition}
\label{def:full gandh defs}
For $x \in P$ define
$$
G(x) =
\max\left\{
\max_{j\,:\, x_0 \le x_j < x} G_j(x),\,
\max_{\substack{(i,j)\\{ x_0 \le x_i < x_j < x}}}  G_{i,j}(x),\, \max_{\substack{(i,j)\\{ x_0 \le x_i < x_j < x}}}  \bar  G_{i,j}(x)
\right\}
$$
$$
H(x)
=
\max\left\{
\max_{i\,:\, x < x_i \le x_n} H_i(x),\, 
\max_{\substack{(i,j)\\{ x <x_i < x_j \le x_n}}}  H_{i,j}(x),\,
\max_{\substack{(i,j)\\{ x < x_i < x_j \le x_n}}} \bar  H_{i,j}(x)
\right\}
$$
\end{definition}

\Cref{thm:reduction2} 
 then  immediately implies the main theoretical result of this paper:
\begin{theorem}
\label{thm:reduction3}
$$R_{\max}(P,x) =
\max \left\{
G(x),\, H(x)
\right\}.
$$
\end{theorem}

\begin{proof}
By construction, each of the terms comprising $G(x)$ and $H(x)$ are lower bounds on $R_{\max}(P,x).$  This is because, if  $x_t < x$,  for any scenario $s$ and any $y \in P$,
$$\Theta(P,x:s) \ge \Theta_L(P,x:s) \ge g_t(x:s)
\quad
\mbox{and}
\quad 
\Theta_{OPT}(P:s) \le \Theta(P,y:s)
$$
so
$$R_{\max}(P,x) \ge R(P,x:s) = \Theta (P,x:s)-\Theta_{OPT}(P:s) \ge g_t(x:s) - \Theta(P,y:s).$$
Plugging into the formulas given in \Cref{def:gandh defs,def:full gandh defs}   immediately shows that $R_{\max}(P,x) \ge G(x).$  The proof that $R_{\max}(P,x) \ge H(x)$ is similar.

\Cref{thm:reduction2}  shows that $R_{\max}(P,x)$ is achieved by one of the cases that $G(x)$ and $H(x)$ are enumerating, proving correctness.
\end{proof}

\paragraph{ The minimization ranges in \Cref{def:gandh defs}}
We note  that the definitions of $G_j$ and $G_{i,j}$  could  be modified without  affecting validity.  More explicitly, the range ``$\yjyx$'' in their inner minimizations could be replaced by
``${y\,:\, x_j \le y \le x}$'' to more closely mirror the statement of \Cref{thm:reduction2}.  \Cref{thm:reduction3} would remain correct under this replacement.  The longer ranges are used to simplify the statements of the later evaluation procedures. A similar replacement of    ``$\yxyi$'' with ``${y\,:\, x \le y \le x_i}$'' could be made in the definitions of  $H_i$ and $H_{i,j}.$

\paragraph{Comparison to the uniform capacity case}
It is instructive at this point to compare the result above to what is known in the uniform capacity case. 

For the uniform capacity case on paths,  it is known \cite{ChengHKNSX13,Wang2014b,Higashikawa2014a,Bhattacharya2015} that the set of 
  $s_{0,0}(w_0^-,w_0^-)$ ( i.e., all $w_i^-$s) and the scenarios  $s_{0,j}(w_0^+,w_j^+)$ or  $s_{j,n}(w_j^+,w_n^+)$ for some $j$ provide the worst-case scenarios for all the sinks.  This implies the existence of a simple $O(n)$ sized set of worst-case scenarios, all structurally  independent of the actual input values. The existence of this set  is the cornerstone of the fast (best case $O(n)$ \cite{Bhattacharya2015}) algorithms  for this problem.  Similar structural results hold for the worst-case scenarios for the  uniform-capacity  minmax regret $k$-sink on a  path \cite{arumugam2019optimal} problem and one sink on a tree problem  \cite{Bhattacharya2015}.

In contrast, in the general capacities case, no such simple finite set of worst-case scenarios seem to exist.   
\Cref{thm:reduction2} reduces the search space of worst-case scenarios substantially, but not to  a finite set. 
 
 We now study some properties of the functions $G(x)$ and $H(x)$ that will be useful later.
 \begin{lemma}\ 
 \label{lem:GHcont}
 \begin{enumerate}
 \item  Let $x \in \left(x_u, x_{u+1}\right).$ Then
\begin{equation}
\label{eq:GHcontd}
\quad G(x)   =   G(x_{u+1})- (x_{u+1} - x), 
 \quad\mbox{and}\quad 
 H(x)  =  H(x_u) - (x- x_u).
\end{equation}
 \item  Let $0 \le u < n.$  Then
 $$ 
 G(x_u)    \le  G(x_{u+1})- (x_{u+1} - x_u),
 \quad\mbox{and}\quad 
  H(x_{u+1}) \le  H(x_u) - (x_{u+1}- x_u).
 $$
  \end{enumerate}
 \end{lemma}
\begin{proof}
Let $x \in (x_u,x_{u+1})$ and $s \in S_{i j}.$  From  \Cref{lem:evac,cor:smooth evac},
$$
\forall t\leq u,\quad g_t(x:s) = g_t(x_{u+1}:s) - (x_{u+1}-x)
\quad\mbox {and}\quad
\forall i \geq u+1, \quad h_i(x:s) = h_i(x_{u}:s) - (x_{u+1}-x).
$$
Directly plugging in the formulas from \Cref{def:gandh defs}  yields,
for all  $x \in (x_u, x_{u+1}),$ 
\begin{align*}
G_j(x) &= G_j(x_{u+1})- (x_{u+1} - x), &  H_i(x) &=H_i(x_u) - (x- x_u), \\ 
G_{i,j}(x) &= G_{i,j}(x_{u+1})- (x_{u+1} - x), &  H_{i,j}(x) &=H_{ i,j}(x_u) - (x- x_u),\\ 
\bar G_{i,j}(x) &= \bar G_{i,j}(x_{u+1})- (x_{u+1} - x), &  \bar H_{i,j}(x) &=\bar H_{ i,j}(x_u) - (x- x_u). 
\end{align*}
(1.) then follows from the definition of  $G(x)$ and $H(x)$.

Similarly, from~\Cref{lem:evac,cor:smooth evac}, for every $s \in S_{i,j}$, we get 
\[
\forall t < u, \quad g_t(x_u:s) = g_t(x_{u+1}:s) -(x_{u+1}-x_u) \quad \text{and} \quad \forall i > u+1,\quad h_i(x_{u+1}) = h_i(x_u) -(x_{u+1}-x_u)
\]
Plugging these formulas into the definitions again yields

\begin{eqnarray*}
G(x_u) 
&=&
\max\left\{
\max_{j\,:\, x_0 \le x_j < x_u} G_j(x_{u+1}),\,
\max_{\substack{(i,j)\\{ x_0 \le x_i < x_j < x_u}}}  G_{i,j}(x_{u+1}),\, \max_{\substack{(i,j)\\{ x_0 \le x_i < x_j < x_u}}}  \bar  G_{i,j}(x_{u+1})
\right\}-(x_{u+1}-x_u)\\
&\le&
\max\left\{
\max_{j\,:\, x_0 \le x_j < x_{u+1}} G_j(x_{u+1}),\,
\max_{\substack{(i,j)\\{ x_0 \le x_i < x_j < x_{u+1}}}}  G_{i,j}(x_{u+1}),\, \max_{\substack{(i,j)\\{ x_0 \le x_i < x_j  <x_{u+1}}}}  \bar  G_{i,j}(x_{u+1})
\right\}-(x_{u+1}-x_u)\\
&=&G(x_{u+1}) -(x_{u+1}-x_u),
\end{eqnarray*}
proving the left side of (2.).
Note that the difference between the first and second lines is the extension of the ranges on which the maximum is taken.

The proof of the right side of (2.),  that  $ H(x_{u+1}) \le  H(x_u) - (x_{u+1}- x_u),$ is similar.
\end{proof}

The next section uses the structural information provided by this section  to derive a
$O(n^4)$  
procedure for finding a worst case scenario for $x.$

%% file: real_algorithm_v2.tex
\section{The Algorithm}
\label{sec: real alg}



\Cref{thm:reduction3,lem:GHcont}
 permit   efficient calculation of $R_{OPT}(P)$. 
\begin{lemma}
\label{lem:Un}  Let $U(n)$ be the time required to calculate $H(x_i)$ and $G(x_i)$ for any $x_i \in P.$  Then
 $R_{OPT}(P)$ can be calculated in $O(U(n) \log n)$ time.
\end{lemma}
\begin{proof}
First note that, from  \Cref{thm:reduction3}, 
$$R_{\max}(P,x_i) =
\max \left\{G(x_i),\, H(x_i)
\right\}
$$
can be calculated in $O(U(n))$ time for any $x_i \in P.$

Let $x^*$ be the sink location such that $R_{OPT}(P) =R_{\max}(P,x^*)$.
From \Cref{lem:rmax unimodal}, $R_{\max}(P,x)$ is a unimodal function, so, from \Cref{obs:bs unimodal},  the unique index $u$  satisfying 
 $x_u \le x^* \le  x_{u+1}$ can be found 
using $O(\log n)$ queries, where  a query is the evaluation of $R_{\max}(P,x_i)$ for some $x_i \in P.$
Thus, 
 $u$ can be found  in  $O(U(n)  \log n)$ time.

\medskip

After the conclusion of this process,   $u,$  $G(x_u),$  $H(x_u)$,
$R_{\max}(P,x_u)$,  $G(x_{u+1})$,  $H(x_{u+1})$  and $R_{\max}(P,x_{u+1})$ are all known.

For $x \in \left(x_u,x_{u+1}\right)$,  \Cref{lem:GHcont}(1) then yields
$$
G(x)  =\  G(x_{u+1})- (d_u - \delta)
\quad\mbox{and}\quad
H(x) =  H(x_u) - \delta
$$
where $\delta = x-x_u$ and 
so, from 
 \Cref{thm:reduction3},
 $$R_{\max}(P,x) =\max  \left\{   G(x_{u+1}) - (x_{u+1} -x),\,  H(x_u) -(x - x_u)  \right\}.
$$

\Cref{lem:GHcont}  shows that  $h(x) = R_{\max}(P,x)$  satisfies the conditions of 
\Cref{lem:sidelimits} for  $x \in[\ell,r] =  [x_u,x_{u+1}]$ with $f(x) = G(x)$ and $g(x) = H(x).$  Thus  
\begin{eqnarray*}
 R_{OPT}(P) &=&
\min_{x \in  [x_u,x_{u+1}]} R_{\max}(P,x)\\
&=&
 \min\left(
R_{\max}(P,x_u),\, R_{\max}(P,x_{u+1}) ,\,  \min_{x \in  [x_u,x_{u+1}]} \bar h(x)
\right)
\end{eqnarray*}
where
$$\bar h(x) = \max  \left\{   G(x_{u+1}) - (x_{u+1} -x),\,  H(x_u) -(x - x_u)  \right\}.$$
Since  $\min_{x \in  [x_u,x_{u+1}]} \bar h(x)$ can be computed in an additional $O(1)$ time, this completes the proof.
\end{proof}

\Cref {sec: Main Utility}
 will prove the following 
\begin{theorem}
\label{thm: main utility}
Let $x \in P$,  $x_j \le x$ and  $i < j$.  Then each of
$G_j(x),$ $G_{i,j}(x)$, $\bar G_{i,j}(x)$, $H_i(x),$ $H_{i,j}(x)$, $\bar H_{i,j}(x)$ can be evaluated in $O(n^2)$ time.
\end{theorem}
Together with 
\Cref{def:full gandh defs}   and  \Cref{thm:reduction3},  plugging into
\Cref{lem:Un},
this immediately implies the main result of this paper:

\medskip

\par\noindent{\bf Theorem \ref{thm:nlog4}}
The 1-sink minmax regret location problem with general capacities on paths can be solved in $O(n^4 \log n)$ time.  That is, $R_{OPT}(P)$ can be calculated in $O(n^4 \log n)$ time.

%% file: revised_upper_envs_v1.tex
\section{Upper Envelopes, Good Functions and the Key Technical Lemma}
\label{sec: UE}

The preceding sections developed the combinatorial results needed  to understand the structure of the one-sink minmax regret problem and developed an $O(n^4 \log n)$ time algorithm for solving the problem.  The algorithm's running time (\Cref{thm:nlog4})  depended upon the correctness of  \Cref {thm: main utility}.
 The remainder of the paper will prove  \Cref {thm: main utility}. 

The first part of this section reviews some properties of upper envelopes.
The second part uses those to prove \Cref{lem: fork H cor}, the key technical lemma.
\Cref{sec: Main Utility} explains how \Cref{lem: fork H cor} implies \Cref {thm: main utility}.  

\subsection{Properties of Upper Envelopes}

\begin{definition}\   

\begin{itemize}
\item For $1 \le j \le t$, let $\ell_j, r_j$  be such that $\ell_j \le r_j$ and for  $j < t$, $r_j \le \ell_{j+1}.$\\
Let  $I_j$ be one of the four intervals $[\ell_i,r_i],$  $(\ell_i,r_i],$ $[\ell_i,r_i),$ or $(\ell_i,r_i).$\\
 A function $f(x)$ will be called {\em piecewise-linear}  of size $t$ if the domain  of $x$ is $\cup_{j=1}^t I_j$ where 
$$\forall x \in  I_j,\quad 
f(x) = m_j x + b_j
$$
for some $m_j,b_j.$
The points $\bigcup_j \{\ell_j,r_j\}$ are the {\em critical points} of $f$.

{\em \small Note: if $f(x)$ is piecewise linear, the definition implies that that if $r_j = \ell_{j+1}$ and $I_j$ is closed from the right and $I_{j+1}$ is closed from the left, then
$ m_j r_j + b_j = m_{j+1} r_j + b_{j+1}.$
In particular, if all of the $I_j$ are closed and   $ \forall j < n,$   $r_j = \ell_{j+1}$,  then $f$ is a continuous piecewise-linear function. 
}

\item Let 
$f_j(x) = m_j x + b_j$, $1 \le j \le t$  be  a set of lines.
Their {\em upper envelope} is the function
$$f(x) = \max_{1 \le j \le t} f_j(x).$$
The set of lines is {\em sorted} if  $0 \le m_ 1 \le m_2 \le \cdots \le m_t$.  
\end{itemize}

For simplicity, the sequel  will  use the following terminology:
\begin{itemize}
\item $f(x)$ is   a {\em good function}, if it is a continuous piecewise linear function of size $O(n)$  with all the slopes $m_j \ge 0.$
\item  $f(x)$ is a  {\em positive function}, if it is a good function with all of the slopes $m_j>0.$
\end{itemize}
Furthermore, we will say that  a piecewise-linear  function restricted to some interval   is  {\em known} if its  critical points and associated linear functions  given in sorted order are known.
 {\em Constructing  a piecewise-linear function} $f$  will mean  knowing $f.$
\end{definition}

\begin{observation}
\label{obs:ue}

Let $f_j(x)$, $j=1,\ldots,m$  be a sorted  set  of lines  and $f(x)$ be 
  their upper envelope.
\begin{itemize}
\item  Then $f(x)$ is a  continuous piecewise-linear  function of size at most $m$ over the reals.
\item Let $I=[\ell,r]$ be an interval.    
Then there exists a sequence of $m' \le m$ indices,
$ 1 \le i_1 < i_2 < \cdots < i_{m'}  \le  m$ and a sequence of  $m'+1$  critical points $\ell= q_0 < q_1  < \cdots <q_{m'} = r$ such that
$$\forall x \in  [q_{t-1}, q_t],\quad 
f(x) = f_{i_t}(x).
$$
\item  $f(x)$ in $[\ell,r]$ can be constructed in $O(m)$ time.
\end{itemize}
\end{observation}

We will later use the following simple lemma.
\begin{lemma}
\label{lem:maxint}
Let $f(x)$ and $g(x)$ be two known  piecewise-linear functions of size $O(n)$ defined in $[\ell, r].$  Then
$$ \max_{x \in [\ell, r]} (f(x) - g(x))
\quad\mbox{and}\quad
\argmax_{x \in [\ell, r]} (f(x) - g(x))
$$
can be calculated in $O(n)$ time.
\end{lemma}
\begin{proof}
In $O(n)$ time, merge the critical points of the two functions into one sorted list.  These sorted points partition $[\ell,r]$ into $O(n)$ intervals with consecutive critical points as endpoints.  In each of these intervals  $f(x)$ and $g(x)$ are each represented by a single line so 
$\max (f(x) - g(x))$ can be calculated in $O(1)$ time.  Taking the maximum of these $O(n)$ values yields the final answer.
\end{proof}

 \medskip

The following useful facts are straightforward to prove and are collected together for later use.
\begin{lemma}
\label{lem: fork good cons}
Let  $f: I \rightarrow  \Re$,  $g: I' \rightarrow  \Re$ be known piecewise linear functions of size $O(n)$.
\begin{enumerate}
\item If $f: I \rightarrow  \Re$  is  a positive function,   then $f^{-1} : f(I) \rightarrow I$ is also a  positive function and can be constructed in $O(n)$ time.
\item Define    $h =f+g  : I \cap I' \rightarrow \Re$  by  $h(\alpha) = f(\alpha)+g(\alpha)$.  Then $h$  can be constructed in $O(n)$ time. If $f$ and $g$ are both good, then $h$ is also good; if at least one of  $f$ and $g$ are also positive, then $h$ is also positive.
\item Define $h_1,h_2  : I \cup I' \rightarrow \Re$ by
$$
h_1(\alpha) =
\left\{
\begin{array}{ll}
f(\alpha) & \mbox{if $\alpha \in (I \cup I') \setminus I'$}\\
\min(f(\alpha),g(\alpha)) & \mbox{if $\alpha \in (I \cap I')$}\\
g(\alpha) & \mbox{if $\alpha \in (I \cup I') \setminus I$}
\end{array}
\right., 
\quad 
h_2(\alpha) =
\left\{
\begin{array}{ll}
f(\alpha) & \mbox{if $\alpha \in (I \cup I') \setminus I'$}\\
\max(f(\alpha),g(\alpha)) & \mbox{if $\alpha \in (I \cap I')$}\\
g(\alpha) & \mbox{if $\alpha \in (I \cup I') \setminus I$}
\end{array}
\right.
$$
Then $h_1,h_2$ can be constructed in $O(n)$ time.  Furthermore, suppose  $h_1$ (resp $h_2$) is continuous.  Then if  $f$ and  $g$ are both good,  $h_1$ (resp. $h_2$) is also good.  If
$f$ and $g$ are both positive, then $h_1$ (resp. $h_2$) is  also positive.
\item Let  $c_1$ be  a constant and $c_2> 0$ be a constant.  Define $f_1(\alpha) = f(\alpha-c_1)$ and $f_2=c_2 f(\alpha).$  Then, $f_1$ and $f_2$ are good functions that can be constructed in $O(n)$ time.
Furthermore, if $f$ is positive then $f_1$  and $f_2$ are positive.
\end{enumerate}

\end{lemma}

\begin{definition}
\label{def: bdef}
Let $a_1 \le a_2$ and $b_1 \le b_2$  Define
\begin{eqnarray*}
B(a_1,a_2,b_1,b_2) &=&  [a_1,a_2] \times  [b_1,b_2] = \Bigl\{ (\alpha_1,\alpha_2) \,:\,   \alpha_1 \in  [a_1,a_2], \,   \alpha_2  \in [b_1,b_2]\Bigr\},\\
B(a_1,a_2,b_1,b_2: \alpha) &=& \Bigl\{ (\alpha_1,\alpha_2) \in B(a_1,a_2,b_1,b_2) \;:\, \alpha_1 + \alpha_2 = \alpha \Bigr\}.
\end{eqnarray*}
\end{definition}

The two main utility lemmas we will need are given below and proven later in  Section \ref{sec: Utility}.

\begin{lemma}
\label{lem: fork M lem}
Let $f_L : [a_1,a_2] \rightarrow \Re$ and  $f_R : [b_1,b_2] \rightarrow \Re$ be known good positive functions
and set $B=B(a_1,a_2,b_1,b_2)$ and $B(\alpha)=B(a_1,a_2,b_1,b_2: \alpha) $ as introduced in \Cref{def: bdef}.
Define
$$M:  [a_1+b_1, a_2+b_2]  \rightarrow \Re, \quad
 M(\alpha) =
\min_{ (\alpha_1,\alpha_2) \in B(\alpha)} \max\bigl(f_L(\alpha_1),\, f_R(\alpha_2)\bigr).
 $$
 Then  $M(\alpha)$ 
 is a good function that can be constructed in $O(n)$ time.
\end{lemma}

\begin{lemma}
\label{lem: fork H lem}
Let $f_L : [a_1,a_2] \rightarrow \Re$ and  $f_R : [b_1,b_2] \rightarrow \Re$ be known good positive functions.
Furthermore assume  the slope sequences of
$f_L(\alpha)$ and also  $f_R(\alpha)$   
 are monotonically increasing.
Set  $B=B(a_1,a_2,b_1,b_2)$ and $B(\alpha)=B(a_1,a_2,b_1,b_2: \alpha) $ as introduced in \Cref{def: bdef}.
 Finally, let $\ell \le r$ and define
$$M(\alpha) = \min\limits_{\substack{ {(\alpha_1, \alpha_2) \in B(\alpha)} \\ { y \in [\ell,r]}}}
 \max\left(  f_L(\alpha_1) + y,  f_R(\alpha_2) - y\right).
 $$
Then 
 $$M: \left[a_1+b_1,a_2+b_2\right] \rightarrow \Re^+$$ is a good function that can be constructed in $O(n)$ time.
\end{lemma}

\subsection{The Key Technical Lemma}

We now use the properties of upper envelopes introduced in the previous subsection to prove~\Cref{lem: fork H cor}, the key technical lemma of the paper.

First, we start with defining the upper envelope functions that underlie the sink evacuation problem. 
\begin{definition}
\label{def:UE LE}
Let $i,j$  be indices   and $s$ any  scenario. 
For an index $t$, set 
 $d_{t,j} = |x_j - x_t|.$
Define
\begin{align*}
\lue_{i,j}(\alpha:s) &=  \Theta_L(P,x_j :s_{-i} (\alpha)) &= \max_{0 \le t <j}  g_t(x_j:s_{-i} (\alpha))   &= \max_{0 \le t <j}  \left( d_{t,j}  +  \frac 1 {c(x_t,x_j)}  W_{0,t}\left(s_{-i}(\alpha)\right) \right), \\
\rue_{i,j}(\alpha:s) &=  \Theta_R(P,x_j :s_{-i} (\alpha))&=  \max_{j < t \le n}  h_t(x_j:s_{-i} (\alpha))      &=  \max_{j < t \le n}  \left(  d_{j,t}+  \frac 1 {c(x_j,x_t)}W_{t,n}\left(s_{-i}(\alpha)\right)\right).
\end{align*}
\end{definition}

\begin{observation}
\label{obs:LUE}
For   
fixed indices  $ t\le j$  and scenario $s,$  
$$
W_{t,j}\left(s_{-i}(\alpha\right))
= \sum_{k=t}^j w_k\left(s_{-i}(\alpha) \right)=
\left\{
\begin{array}{ll}
W_{t,j}(s) & \mbox{if  $i< t$ or $j < i$}\\
W_{t,j}(s)  + \alpha - w_i(s) & \mbox{if $t \le i \le j$ }
\end{array}
\right..
$$

We now note that    the evacuation functions of $s_{-i}(\alpha)$ are upper envelopes of lines in $\alpha.$
%
 \end{observation}

\begin{lemma}
\label{lem:env} Let $s$ be any fixed scenario.
Let $x$ be  any fixed sink location. 

Then 
$\Theta_L(P,x: s_{-i} (\alpha))$,  $\Theta_R(P,x: s_{-i} (\alpha))$ and $\Theta (P,x: s_{-i} (\alpha))$ are all, as a function of $\alpha$, upper envelopes of a set of lines with nonnegative slope.  
These functions all have $O(n)$ critical points. Furthermore,  these critical  points and the 
 line equations of the upper envelopes can be calculated in $O(n)$ time.

\end{lemma}
\begin{proof}
From  \Cref{def:UE LE,obs:LUE,cor:smooth evac},
$\Theta_L(P,x: s_{-i} (\alpha))$ and  $\Theta_R(P,x: s_{-i} (\alpha))$  are upper envelopes of $O(n)$ lines and therefore each have $O(n)$ critical points.  
The fact that the envelopes and critical points can be calculated in $O(n)$ time follows directly from \Cref{obs:ue} and  the fact that, for fixed $i,j$, in the definition of
$\lue_{i,j}(\alpha:s)$  ($\rue_{i,j}(\alpha:s) $), the slopes $\frac{1}{c(x_t,x_j)}$ ($\frac{1}{c(x_j,x_t)}$) appear in nondecreasing order as $t$ decreases (increases).

Since the maximum of two upper envelopes with nonnegative slopes  is an upper envelope with nonnegative slope,
$$\Theta(P,x: s_{-i} (\alpha)) = \max\bigl( \Theta_L(P,x: s_{-i} (\alpha)),\,  \Theta_R(P,x: s_{-i} (\alpha)) \bigr)$$
is also an upper envelope with nonnegative slope. It can be constructed in $O(n)$ further time through a simple merge of the left and right upper envelopes.
\end{proof}

\Cref{lem:env}  implies that all  three functions are {\em good functions}. They are not necessarily {\em positive} functions because it is possible that they might be constant. It is also possible that for small enough $\alpha$, the functions are constant, but, after passing some  threshold value of $\alpha$, they are monotonically increasing.

The tools  above  enable proving  further technical lemmas that will be needed.

\begin{lemma}
\label{cor: fork M_k cor}
Let $s$ be fixed and $S_{i,j}$  be as introduced  in \Cref{def:sstuff}. Let $k$ satisfy  $i \le k \le j$ and 
set 
$$B= B( a_1,a_2,b_1,b_2),\quad  B(\alpha) = \left\{ (\alpha_1,\alpha_2) \in B \;:\, \alpha_1 + \alpha_2 = \alpha \right\}$$
as introduced in \Cref{def: bdef}. Define.
\begin{equation}
\label{eq: fork 27}
M_k(\alpha) =
\min_{ (\alpha_1,\alpha_2) \in B(\alpha)}
 \Theta(P,x_k:s_{i,j}(\alpha_1,\alpha_2))
 \end{equation}
 Then 
 $$M_k(\alpha): \left[a_1+b_1,\,a_2+b_2\right] \rightarrow \Re^+$$ is a good function that can be constructed in $O(n)$ time.
\end{lemma}

\begin{proof}  
Set $s = s_{i,j}(0,0)$ and  
$$
\begin{array}{ccccc}
f_L(\alpha_1) &=&  \Theta_L\left(P,x_k:s_{i,j}(\alpha_1,\alpha_2)\right) &=& \lue_{i,k}(\alpha_1:s),\\
f_R(\alpha_2) &=& \Theta_R\left(P,x_k:s_{i,j}(\alpha_1,\alpha_2)\right) & =&\rue_{j,k}(\alpha_2:s),
\end{array}
$$
where the second equality on each line come from the fact that $i \le k \le j$.
Then
\begin{equation}
\label{eq:fork1}
 \Theta(P,x_k:s_{i,j}(\alpha_1,\alpha_2)) = \max(f_L(\alpha_1), f_R(\alpha_2)).
\end{equation}

From \Cref{def:UE LE} and \Cref{lem:env},
$\lue_{i,k}(\alpha_1:s)$ and $\rue_{j,k}(\alpha_2:s)$ are both good positive functions. The proof   follows immediately by applying 
\Cref {lem: fork M lem}.
\end{proof}

\begin{lemma}
\label{lem: fork H cor}
Let $s_{i,j}(\alpha,\beta)$ be  as introduced  in \Cref{def:sstuff}. Let $k$ satisfy  $i \le k < j$ and 
set 
$$B= B( a_1,a_2,b_1,b_2),\quad  B(\alpha) = \left\{ (\alpha_1,\alpha_2) \in B \;:\, \alpha_1 + \alpha_2 = \alpha \right\}.$$
as introduced in \Cref{def: bdef}. Define
\begin{equation}
\label{eq: fork 35}
M^{(k)}_{i,j}(P:\alpha) = \min\limits_{\substack{ {(\alpha_1, \alpha_2) \in B(\alpha)} \\ { y \in \left[ x_k,x_{k+1}\right]}}}
 \Theta(P,y:s_{i,j}(\alpha_1,\alpha_2)).
 \end{equation}
 Then 
 $$M^{(k)}_{i,j}(P:\alpha): \left[a_1+b_1,\,a_2+b_2 \right] \rightarrow \Re^+$$ is a good function that can be constructed in $O(n)$ time.
\end{lemma}
\begin{proof}   
Let $y \in (x_k,x_{k+1})$.  Recall, 
from \Cref{cor:smooth evac},
$$
\begin{array}{cclcl}
\Theta_L\left(P,y:s_{i,j}(\alpha_1,\alpha_2)\right) &= &  \Theta_L\left(P,x_{k+1}:s_{i,j}(\alpha_1,\alpha_2)\right) - (x_{k+1}-y) &=&  \lue_{i,k+1}(\alpha_1:s) - (x_{k+1}-y), \\
\Theta_R\left(P,y:s_{i,j}(\alpha_1,\alpha_2)\right) &=&   \Theta_R\left(P,x_{k}:s_{i,j}(\alpha_1,\alpha_2)\right)  - (y- x_k) &=& \rue_{j,k}(\alpha_2:s) - (y- x_k).
\end{array}
$$
Set
$$
f_L(\alpha)  =  \lue_{i,k+1}(\alpha:s) - x_{k+1},\quad
f_R(\alpha) =  \rue_{i,k}(\alpha:s) + x_{k}.
$$
This  permits writing
\begin{eqnarray*}
\Theta_L\left(P,y:s_{i,j}(\alpha_1,\alpha_2)\right) &= &  f_L(\alpha_1) + y,\\
\Theta_R\left(P,y:s_{i,j}(\alpha_1,\alpha_2)\right) &=&    f_R(\alpha_2) - y.
\end{eqnarray*}
Since $\lue_{i,k+1}(\alpha:s)$  and $\rue_{i,k}(\alpha:s)$  are  known good positive functions in $\alpha$,  
$f_L(\alpha)$ and $f_R(\alpha)$ are also good positive functions and can be constructed in $O(n)$ time.

Now, for $\forall y \in [x_k,x_{k+1}]$ (note that this is a {\em closed} interval), define
$$C(y, \alpha_1,\alpha_2)
 = 
\max\left(  f_L(\alpha_1) + y,  f_R(\alpha_2) - y\right).
$$
By definition
\begin{equation}
\label{eq:BB:1}
\forall y \in (x_k,x_{k+1}),\, (\alpha_1,\alpha_2) \in  B,\quad
C(y, \alpha_1,\alpha_2)  =
 \Theta\left(P,y:s_{i,j}(\alpha_1,\alpha_2)\right).
\end{equation}

Because $C$ is piecewise linear, it  is uniformly continuous in $ [x_k,x_{k+1}] \times B$ and thus,
by  the compactness of  $ [x_k,x_{k+1}] \times B,$
\begin{equation}
\label{eq: fork 53}
\forall \alpha, \quad  D(\alpha) 
= \min\limits_{\substack{ {(\alpha_1, \alpha_2) \in B(\alpha)} \\{y \in [x_k,x_{k+1}]}}} C(y,\alpha_1,\alpha_2)
= \inf\limits_{\substack{ {(\alpha_1, \alpha_2) \in B(\alpha)} \\{y \in (x_k,x_{k+1})}}} C(y,\alpha_1,\alpha_2),
\end{equation}
where $D(\alpha)$ exists and is continuous. 

\Cref{cor:techlimit} and  \Cref{eq:BB:1}
immediately imply that for any fixed $\alpha_1,\alpha_2,$ 
$$
\hspace*{-.2in} \min_{y \in [x_k,x_{k+1}]}  \Theta(P,y:s_{i,j}(\alpha_1,\alpha_2))) =
\min
\small \left(
\Theta(P,x_k:s_{i,j}(\alpha_1,\alpha_2)),\, 
 \Theta(P,x_{k+1}:s_{i,j}(\alpha_1,\alpha_2)),\,
\min_{y \in [x_k,x_{k+1}]} C(y,\alpha_1,\alpha_2)
 \right).
 $$

Now fixing $\alpha$ and taking the minimum of both sides over all $(\alpha_1,\alpha_2) \in B(\alpha)$ yields

\begin{align*}
 \label{eq: fork 60}
M_{i,j}^{(k)}(P:\alpha) &=&&
 \min\limits_{\substack{ {(\alpha_1, \alpha_2) \in B(\alpha)} \\ { y \in \left[ x_k,x_{k+1}\right]}}}
 \Theta(P,y:s_{i,j}(\alpha_1,\alpha_2))\\
 &=&&
 \min
 \left(
  M_{k}(P:\alpha),
\, M_{k+1}(P:\alpha),\,
\min \limits_{\substack{ {(\alpha_1, \alpha_2) \in B(\alpha)} \\{y \in (x_k,x_{k+1})}}} C(y,\alpha_1,\alpha_2)
 \right)\\
  &=&&
 \min
 \left(
  M_{k}(P:\alpha),
\, M_{k+1}(P:\alpha),\,
D(\alpha)
 \right).
\end{align*}

\Cref{cor: fork M_k cor} already states that $M_k(P:\alpha)$ and $M_{k+1}(P:\alpha)$ are good functions that can be constructed  in $O(n)$ time.
\Cref{lem: fork H lem} and the definition of $C$  show that 
 $D(\alpha)$ is also a  good function  that can be constructed  in $O(n)$ time. 
\Cref{lem: fork good cons} (3) then immediately implies that $M^{(k)}(P:\alpha)$  is also  a  good function  that can be constructed  in $O(n)$ time, proving the lemma.  
\end{proof}

The lemma has  the following useful Corollary:
\begin{corollary}
\label{lem: fork H cor2}
Let $a_1 \le a_2$ and $ s$ be any fixed scenario and  $0 \le j,k \le n.$  Then 
\begin{equation}
\label{eq: fork 40}
M^{(k)}_{j}(P:\alpha) = \min_{y \in \left[ x_k,x_{k+1}\right]}
 \Theta(P,y:s_{-j}(\alpha))
 \end{equation}
 is a good function over the interval $[a_1,a_2]$ that can be constructed in $O(n)$ time.
\end{corollary}
\begin{proof}
Without loss of generality assume $k < j$  (the other direction is symmetric), and choose any index $i \le k.$
Note that
$s_{-j}(\alpha) = s_{i,j}\left(w_i(s),\alpha \right).$
Then 
$$M^{(k)}_{j}(P:\alpha) = M^{(k)}_{i,j}(P:\alpha) 
$$
where $(a_1,a_2,b_1,b_2)= \left( w_i(s),\, w_i(s), a_1, a_2 \right)$, so the proof follows from Lemma \ref{lem: fork H cor}.
\end{proof}

%% file: main_utility.tex
\section{The proof of \Cref{thm: main utility}}
\label{sec: Main Utility}

This section shows how  \Cref{lem: fork H cor} permits evaluating  each of the 6 functions in \Cref{thm: main utility}  in $O(n^2)$ time, proving 
\Cref{thm: main utility}.

Recall the definition of  $s_{i,j}(\alpha_1,\alpha_2)$ from \Cref{def:sstuff}. Set $\bar s=s_{i,j}(0,0).$
Note that if $x_i < x_j \le x_t < x$ then $W_{0,t}(s_{i,j}(\alpha_1,\alpha_2)) = W_{0,t}(\bar s) + \alpha_1 + \alpha_j.$ Thus 
\begin{eqnarray}
g_t(x:s_{i,j}(\alpha_1,\alpha_2))
&=& d(x_t,x)  +  \frac 1 {c(x_t,x)}  W_{0,t}(s_{i,j}(\alpha_1,\alpha_2)) \nonumber \\
&=&  \left(d(x_t,x)  +  \frac 1 {c(x_t,x)}  W_{0,t}(\bar s)\right) +  \frac 1 {c(x_t,x)} (\alpha_1+\alpha_2) \label{eq:gt linear}
\end{eqnarray}
is a linear function in $\alpha = \alpha_1 + \alpha_2.$  Also note that this function is well-defined for all $\alpha \ge 0.$ 

We now go through the first three functions, one by one.

\begin{lemma}[Evaluation of $\bar G_{i,j}(x)$] 
\label{lem: barGij eval}
 Fix $0 \leq i < j \leq n$ and $x \in [x_0,x_n]$ such that  $x_j < x.$  Then
 \begin{enumerate}
 \item 
 $\bar G_{i,j}(x)  = \max_{i \le u < j} \bar G^{(u)}_{i,j}(x)$
 where 
 $$\bar G^{(u)}_{i,j}(x)= 
 \max_{ s \in S_{i,j}}
  \left\{ 
   \left(\max_{t\,:\, x_j \le x_t  <x} g_t(x:s)\right) - \min_{y: x_u \le y \le x_{u+1}} \Theta(P,y:s) 
  \right\}.$$
  \item  For all $u,$  $i \le u < j,$ $\bar G^{(u)}_{i,j}(x)$ can be evaluated in $O(n)$ time.
  \item$ \bar G_{i,j}(x) $ can be evaluated in $O(n^2)$ time.
  \end{enumerate}
\end{lemma}
\begin{proof}
 (1)  follows from a simple manipulation of the definition of $\bar G_{i,j}(x).$

For (2), note that from Equation (\ref{eq:gt linear}),  
$$\max_{t\,:\, x_j \le x_t < x} g_t(x:s_{i,j}(\alpha_1,\alpha_2))
= F_{i,j}(\alpha_1+ \alpha_2)
$$
where
$$
F_{i,j}(\alpha) =\max_{t\,:\, x_j \le x_t < x}
\left\{\left(d(x_t,x)  +  \frac 1 {c(x_t,x)}  W_{0,t}(\bar s)\right) +  \frac 1 {c(x_t,x)}   \alpha\right\}.
$$

Since $F_{i,j}(\alpha)$ is the upper envelope of $O(n)$ lines given by increasing slope, it is a good function. Furthermore, by \Cref{obs:ue},  it can be constructed in $O(n)$ time.

Now  let $M^{(u)}_{i,j}(P:\alpha)$ be  as introduced in \Cref{lem: fork H cor}  with
$(a_1,a_2,b_1,b_2)= \left(w^-_i,  w^+_i,  w^-_j, w^+_j\right)$
and $R = \left[
w^-_i + w^-_j,w^+_i+ w^+_j
 \right].$
Then
\begin{eqnarray*}
\bar G^{(u)}_{i,j}(x) &=& \max_{ s \in S_{i,j}} \left\{ \left(\max_{t\,:\, x_j \le x_t  <x} g_t(x:s)\right) - \min_{y\,:\, x_u \le y \le x_{u+1}} \Theta(P,y:s) \right\}\\
&=& \max_{ (\alpha_1,\alpha_2) \in B} \left\{ \left(\max_{t\,:\, x_j \le x_t < x} g_t(x:s_{i,j}(\alpha_1,\alpha_2))\right) - \min_{y\,:\, x_u \le y \le x_{u+1}} \Theta(P,y:s_{i,j}(\alpha_1,\alpha_2)) \right\}\\
&=& \max_{ (\alpha_1,\alpha_2) \in B} 
\left\{ 
F_{i,j}(\alpha_1 + \alpha_2)
- \min_{y\,:\, x_u \le y \le x_{u+1}} \Theta(P,y:s_{i,j}(\alpha_1,\alpha_2))
 \right\}\\
 &=&
 \max_{\alpha \in R}
 \left(\max_{ (\alpha_1,\alpha_2) \in B(\alpha)} 
\left\{ 
F_{i,j}(\alpha)
- \min_{y\,:\, x_u \le y \le x_{u+1}} \Theta(P,y:s_{i,j}(\alpha_1,\alpha_2))
 \right\}
 \right)\\
 &=&
 \max_{\alpha \in R}
 \left(
F_{i,j}(\alpha)
-  \min\limits_{\substack{ {(\alpha_1, \alpha_2) \in B(\alpha)} \\ { y \in \left[ x_u,x_{u+1}\right]}}}
 \Theta(P,y:s_{i,j}(\alpha_1,\alpha_2))
 \right)\\
 &=& 
 \max_{\alpha \in R}
 \left(
F_{i,j}(\alpha)
-M^{(u)}_{i,j}(P:\alpha)
\right).
\end{eqnarray*}

$F_{i,j}(\alpha)$ (as noted above) and $M^{(u)}_{i,j}(P:\alpha)$ (from \Cref{lem: fork H cor}) are both good functions that can be constructed in $O(n)$ time.
 From \Cref{lem:maxint}, $ \bar G^{(u)}_{i,j}(x)$ can then be calculated in $O(n)$ additional time. This completes the proof of (2).

(3) follows directly from (1) and (2).
\end{proof}

\begin{lemma}[Evaluation of $ G_{i,j}(x)$]
\label{lem: Gij eval}
 Fix $0 \leq i < j \leq n$ and $x \in [x_0,x_n]$ such that  $x_j < x.$  Then  
 \begin{enumerate}
 \item 
 $G_{i,j}(x)  = \max_{j \le u < n} G^{(u)}_{i,j}(x)$
 where 
 $$
 G^{u}_{i,j} (x) = \max_{ s \in S_{i,j}, w_j(s) = w_j^+} \left\{ g_j(x:s) - \min_{y\,:\, x_u \le y \le x_{u+1}}  \Theta(P,y:s) \right\}
 $$
  \item  For all $u,$  $j \le u < n,$ $G^{(u)}_{i,j}(x)$ can be evaluated in $O(n)$ time.
  \item$  G_{i,j}(x) $ can be evaluated in $O(n^2)$ time.
  \end{enumerate}
\end{lemma}

%
%
%
%
%
%
\begin{proof}
 (1)  follows from a simple manipulation of the definition of $ G_{i,j}(x).$

For (2), set  $ s'=s_{i,j}(0,w_j^+).$ Note that $s_{i,j}(\alpha,w^+_j) = s'_{-i}(\alpha).$  Next, set 
$$ h_j(\alpha)  = g_j(x:s_{i,j}(\alpha,w_j^+))  = g_j(x:s'_{-i}(\alpha))
= \left(d(x_j,x)  +  \frac 1 {c(x_j,x)}  W_{0,j}(s')\right) +  \frac 1 {c(x_j,x)} \alpha.
$$
Let $M^{(u)}_{i} (P:\alpha)$ be  as introduced in 
\Cref {lem: fork H cor2} with
$(a_1,a_2)= \left(w^-_i,  w^+_i\right).$
Then
\begin{eqnarray*}
G^{(u)}_{i,j} (x) &=& \max_{ s \in S_{i,j}, w_j(s) = w_j^+} \left\{ g_j(x:s) - \min_{y\,:\, x_u \le y \le x_{u+1}}  \Theta(P,y:s) \right\}\\
&=& \max_{w_i^- \le \alpha \le w_i^+}
\left( 
g_j(x,:\bar s_{-i}\left(\alpha\right))
- \min_{y\,:\, x_u \le y \le x_{u+1}} \Theta\left(P,y:\bar s_{-i}\left(\alpha\right))\right)
\right)\\
&=& 
\max_{w_i^- \le \alpha \le w_i^+}
\left(
h_j(\alpha)
- M^{(u)}_{i} (P:\alpha)  
\right).
\end{eqnarray*}

From \Cref {lem: fork H cor2}, 
 $M^{(u)}_{i}(P:\alpha)$ is  a  good function that can be constructed in $O(n)$ time.
 From \Cref{lem:maxint}, $ G^{(u)}_{i,j}(x)$ can then be calculated in $O(n)$ additional time. This completes the proof of (2).

(3) follows directly from (1) and (2).
\end{proof}

\begin{lemma}[Evaluation of $G_{j}(x)$] 
\label{lem: Gjeval}  Fix $0 \leq  j \leq n$ and $x \in [x_0,x_n]$ such that  $x_j < x.$  Then
\begin{enumerate}
 \item 
 $G_{j}(x)  = \max_{j \le u < n} G^{(u)}_{j}(x)$
 where 
 $$
 G^{u}_{j} (x) = \max_{ s \in S_{j,j}} \left\{ g_j(x:s) - \min_{y\,:\, x_u \le y \le x_{u+1}}  \Theta(P,y:s) \right\}
 $$
  \item  For all $u,$  $j \le u < n,$ $G^{(u)}_{j}(x)$ can be evaluated in $O(n)$ time.
  \item$  G_{j}(x) $ can be evaluated in $O(n^2)$ time.
  \end{enumerate}
  
\end{lemma}
\begin{proof}
(1)  follows from a simple manipulation of the definition of $ G_{i,j}(x).$

For (2), set  $ s'=s_{j,j}(0,0).$  Recall  that $s_{j,j}(\alpha,\alpha) = s'_{-j}(\alpha).$  Next set 
$$\bar  h_j(\alpha)  = g_j(x:s_{j,j}(\alpha,\alpha))  = g_j(x:s'_{-j}(\alpha))
= \left(d(x_j,x)  +  \frac 1 {c(x_j,x)}  W_{0,j}(s')\right) +  \frac 1 {c(x_j,x)} \alpha.
$$
 Now let $M^{(u)}_{j} (P:\alpha)$ be  as introduced in 
\Cref {lem: fork H cor2} with
$(a_1,a_2)= \left(w^-_j,  w^+_j\right).$
Then

\begin{eqnarray*}
G^{(u)}_{j} (x) &=& \max_{ s \in S_{j,j}} \left\{ g_j(x:s) -  \min_{y\,:\, x_u \le y \le x_{u+1}}  \Theta(P,y:s) \right\}\\
&=& \max_{w_j^- \le \alpha \le w_j^+}
\left( 
g_j(x,:\bar s_{-j}\left(\alpha\right))
-   \min_{y\,:\, x_u \le y \le x_{u+1}} \Theta\left(P,y:\bar s_{-j}\left(\alpha\right))\right)
\right)\\
&=& 
\max_{w_j^- \le \alpha \le w_j^+}
\left(
\bar h_j(\alpha)
-  M^{(u)}_j(P:\alpha)
\right).
\end{eqnarray*}
As in the proof of  \Cref{lem: Gij eval},  from \Cref {lem: fork H cor2}, 
 $M^{(u)}_{j}(P:\alpha)$ is  a  good function that can be constructed in $O(n)$ time.
 From \Cref{lem:maxint}, $ G^{(u)}_{i,j}(x)$ can then be calculated in $O(n)$ further time. 

(3) follows directly from (1) and (2).
\end{proof}

\bigskip

\Cref{lem: barGij eval,lem: Gij eval,lem: Gjeval} say that each of  $\bar G_{i,j}(x), G_{i,j}(x)$ and $G_j(x)$ can be evaluated in $O(n^2)$ time. A totally symmetric argument proves that each of  $\bar H_{i,j}(x), H_{i,j}(x)$ and $H_j(x)$  can also be evaluated in $O(n^2)$ time.  This completes the proof of 
\Cref{thm: main utility}.

%% file: forked_algorithm_v9.tex
\section{The proofs of  \Cref{lem: fork M lem,lem: fork H lem}}
 \label{sec: Utility}


\Cref{sec: real alg} proved  \Cref{thm:nlog4}, the main result of this paper,  assuming the correctness of \Cref {thm: main utility}.
\Cref{sec: Main Utility} proved \Cref {thm: main utility}, assuming the correctness of
\Cref{lem: fork H cor}.
 \Cref{sec: UE} proved \Cref{lem: fork H cor}, assuming the correctness of  
of \Cref{lem: fork M lem,lem: fork H lem}.  

\medskip

This section will prove the correctness of \Cref{lem: fork M lem,lem: fork H lem}.  

\medskip

 Before starting, we note that the complexity in the proof of \Cref{lem: fork H lem} arises from requiring  that the resulting piecewise linear function is size  $O(n)$.  If we were willing to allow an $O(n^2)$ bound, the proof would be much shorter.  This would lead to a $O(n^5 \log n)$ algorithm rather than a $O(n^4 \log n)$ one, though.  We also note that if we were willing to allow an  $O(n^5 \log n)$ algorithm, a variant of  \Cref {thm: main utility} with an $O(n^3)$ construction bound  replacing the $O(n^2)$ one could be derived using  a much simpler (and shorter) 
linear programming approach.  This would lead to $O(n^3)$ time algorithms for evaluating each of the 6 terms in \Cref{thm: main utility}, also leading to an $O(n^5 \log n)$ algorithm. The main 
contribution of this section is reducing 
the $O(n^2)$ time down to $O(n)$ by a more detailed argument, allowing the final  $O(n^4 \log n)$ result.


\subsection{Proof of Lemma \ref{lem: fork M lem}}

\begin{proof}  In what follows, it is assumed that   $\alpha \in [a_1+b_1, a_2+b_2]$.
By the continuity of $f_L$ and $f_R$ and the compactness of $B$, $M(\alpha)$ is well-defined and continuous.
Next, define
$$\tilde B(\alpha)  = \Bigl\{(\alpha_1,\alpha_2) \in B(\alpha) \,:\,  M(\alpha) =\max(f_L(\alpha_1), f_R(\alpha_2)) \Bigr\}.$$
Thus
$$
 M(\alpha) =
\min_{ (\alpha_1,\alpha_2) \in\tilde  B(\alpha)} \max\bigl(f_L(\alpha_1),\, f_R(\alpha_2)\bigr).
 $$

Now define the following five conditions:
$$\begin{array}{cccccccc}
(C_1) &  \alpha_1 = a_1, & \quad & (C_2) &  \alpha_1 = a_2, & \quad \quad &  (C_5) & f_L(\alpha_1) = f_R(\alpha_2).\\
(C_3) &  \alpha_2 = b_1, & \quad & (C_4) &  \alpha_2 = b_2,
\end{array}
$$

A  function $\bar M(\alpha)$ is called a {\em witness for condition $(C_i)$} if
$$
\bar M(\alpha) 
\left\{
\begin{array}{ll}
=M(\alpha)  & \mbox{if there exists $(\alpha_1,\alpha_2) \in \tilde B(\alpha)$ that satisfies condition $(C_i)$}\\
\ge M(\alpha) & \mbox{if there does not exist $(\alpha_1,\alpha_2) \in \tilde B(\alpha)$ that satisfies condition $(C_i)$}
\end{array}
\right.
$$
By default, if there does not exist $(\alpha_1,\alpha_2) \in \tilde B(\alpha)$ that satisfies condition $(C_i)$ and the function $\bar M(\alpha)$ is undefined for $\alpha$, we will assume that $\bar M(\alpha) = \infty$ (so that it is  a trivially a witness).

\medskip

\par\noindent\underline{(i) Claim (*).} For every $\alpha$,   $\exists (\alpha_1,\alpha_2) \in \tilde B(\alpha)$  such that  at least one of conditions  $(C_1)$-$(C_5)$ hold.\\[0.05in]
Suppose by contradiction  there exists  some 
 $\alpha$, such that  for every   pair $(\alpha_1,\alpha_2) \in \tilde B(\alpha)$,
none of  $(C_1)$-$(C_5)$  hold. 

Choose any  $(\alpha_1,\alpha_2) \in \tilde B(\alpha)$.  Because $(C_5)$ does not hold there exists $\Delta > 0$ such that 
$|f_L(\alpha_1) - f_R(\alpha_2)| = \Delta.$ Without loss of generality, assume $f_L(\alpha_1) > f_R(\alpha_2)$  so
$$M(\alpha) =  \max(f_L(\alpha_1), f_R(\alpha_2)) = f_L(\alpha_1).$$

From the continuity  of $f_L(z)$ and $f_R(z)$ and the fact that  $(C_1)$-$(C_4)$ do not hold, there exists $\epsilon >0$ such that
\begin{itemize}
\item  $(\alpha_1 - \epsilon, \alpha_2+ \epsilon) \in B(\alpha)$.
\item $ 0 < f_L(\alpha_1) - f_L(\alpha_1 - \epsilon) < \Delta/2.$
\item $ 0 < f_R(\alpha_2+ \epsilon) - f_R(\alpha_2) < \Delta/2.$
\end{itemize}
But then, $(\alpha_1 - \epsilon, \alpha_2+ \epsilon) \in B(\alpha)$, and from the  monotonicity  of $f_L(z)$ and $f_R(z)$, 
$$  \max(f_L(\alpha_1-\epsilon), f_R(\alpha_2+\epsilon))
= f_L(\alpha_1-\epsilon) < f_L(\alpha_1) = M(\alpha),$$
contradicting the definition of  $M(\alpha)$.  
Thus, Claim (*) holds.

\medskip

\par\noindent\underline{(ii) Examining $(\alpha_1,\alpha_2) \in \tilde B(\alpha)$  for which at least one of conditions  $(C_1)$-$(C_4)$  hold:}\\[0.05in]
If $(\alpha_1,\alpha_2) \in B(\alpha)$  and  $\alpha_i$, $i=1,2$ is fixed then  $\alpha_{3-i} = \alpha - \alpha_i$ is also fixed.  In particular defining $M^{(i)}_k(\alpha)$, $i=1,2,3,4$ as below yields
$$
\begin{array} {ccc}
M^{(1)} (\alpha)  &=& \min\limits_{( a_1,\alpha_2) \in B(\alpha)} \max(f_L(a_1), f_R(\alpha_2)) = \max(f_L(a_1), f_R(\alpha-a_1)),\\
M^{(2)} (\alpha)  &=& \min\limits_{( a_2,\alpha_2) \in B(\alpha)} \max(f_L(a_2), f_R(\alpha_2)) = \max(f_L(a_2), f_R(\alpha-a_2)),\\
 M^{(3)} (\alpha)  &=& \min\limits_{( \alpha_1,b_1) \in B(\alpha)} \max(f_L(\alpha_1), f_R(b_1)) = \max(f_L(\alpha-b_1), f_R(b_1)),\\
 M^{(4)} (\alpha)  &=& \min\limits_{(\alpha_1, b_2) \in B(\alpha)} \max(f_L(\alpha_1), f_R(b_2)) = \max(f_L(\alpha-b_2), f_R(b_2)).
\end{array}
$$

Note  that the ranges of  $M^{(1)},$ $M^{(2)},$ $M^{(3)},$ $M^{(4)},$ are, respectively,
$\left[a_1+b_1, a_1+b_2\right],$
$\left[a_2+b_1, a_2+b_2\right]$,
$\left[a_1+b_1, a_2+b_1\right]$
and
$\left[a_1+b_2, a_2+b_2\right]$.

By construction,  each $M^{(i)}$ is, respectively,  a witness for condition $(C_i).$ 

From \Cref{lem: fork good cons} (3) and (4) each of these $M^{(i)}$, $i=1,2,3,4$  is a good function that can be constructed in $O(n)$ time.  Note that, while good, they might not be positive  since in each case, one of the $f(\alpha),g(\alpha)$ being inserted into the definition in \Cref{lem: fork good cons} (3) is a constant function.

\medskip

\par\noindent\underline{(iii) Examining $(\alpha_1,\alpha_2) \in \tilde B(\alpha)$  for which condition  $(C_5)$  holds:} \\[0.05in]
Further define 
\begin{equation}
\label{eq:fork 2}
M^{(5)} (\alpha)
= 
\min\limits_{\substack{ {(\alpha_1, \alpha_2) \in B(\alpha)} \\{f_L(\alpha_1) =f_R(\alpha_2)}}}
\max(f_L(\alpha_1), f_R(\alpha_2)).
\end{equation}
 
\begin{itemize}
\item From \Cref{lem: fork good cons} (1),  both $f^{-1}_L(t)$ and $f^{-1}_R(t)$ are  good positive functions that can be constructed in $O(n)$ time.
\item Let 
$$\bar I = \Bigl[  \max\left(f_L(a_1), f_R(b_1)  \right),\,  \min\left(f_L(a_2),\, f_R(b_2)\right)   \Bigr]$$ denote the range of  $M^{(5)}$. Note that
\begin{equation}
\label{eq:fork 3}
(\alpha_1, \alpha_2) \in B(\alpha)  \ \mbox{ and } \ f_L(\alpha_1) = f_R(\alpha_2) =t
\quad \Longleftrightarrow  \quad
\exists t \in \bar I \ \  \mbox{s.\@t.} \ \ 
f^{-1}_L(t) + f^{-1}_R(t) = \alpha \  
\end{equation}
\item  Set $g(t) = f^{-1}_L(t) + f^{-1}_R(t)$. From   \Cref{lem: fork good cons} (2), $g(t)$ is a good positive function that can be constructed in $O(n)$ time.  
 Set $h(\alpha)= g^{-1}(\alpha)$.
From   \Cref{lem: fork good cons} (1), $h(\alpha)$ is also  a good positive function that can be constructed in $O(n)$ time. 
\item  \Cref {eq:fork 3}  then implies
\begin{equation}
\label{eq:fork 8}
(\alpha_1, \alpha_2) \in B(\alpha)  \ \mbox{ and } \ f_L(\alpha_1) = f_R(\alpha_2) 
\quad \Longleftrightarrow  \quad
\ f_L(\alpha_1) = f_R(\alpha_2)  = h(\alpha).
\end{equation}
\end{itemize}
The facts above imply
\begin{equation}
\label{eq:fork4}
M^{(5)}_k \,:\,  g\left(\bar I\right) \rightarrow \bar I,\quad  
M^{(5)}_k (\alpha)
= 
\min\limits_{\substack{ {(\alpha_1, \alpha_2) \in B(\alpha)} \\{f_L(\alpha_1) =f_R(\alpha_2)}}}
\max(f_L(\alpha_1), f_R(\alpha_2))
=h(\alpha)
\end{equation}
is a good positive function that can be constructed in $O(n)$ time. 
Furthermore,  $M^{(5)}$ is a witness to condition $(C_5)$.

\medskip

\par\noindent\underline{(iv) Completing the proof}\\
From Claim (*)  every $(\alpha_1,\beta_1) \in \tilde B(\alpha)$ must satisfy at least one condition $(C_i)$, $i=1,2,3,4,5,$.
We have seen that each $ M^{(i)}$, $i=1,2,3,4,5$ is a witness to condition $C_i.$ Thus
\begin{equation}
\label{eq:fork 5}
M(\alpha) =
\min_{1 \le i \le 5}
M^{(i)}(\alpha).
\end{equation}

Furthermore, the $M^{(i)}$, $i=1,2,3,4,5$ are all  good functions (with different domains)  that can be constructed in $O(n)$ time.
Since $M(\alpha)$ is continuous,  from \Cref{lem: fork good cons} (3), 
$M(\alpha)$ is also a good function that can be constructed in $O(n)$ time.
\end{proof}

\subsection{Proof of Lemma \ref {lem: fork H lem}}

\begin{proof}
See \Cref{fig:critvert}.

 In what follows, it is assumed that   $\alpha \in [a_1+b_1, a_2+b_2]$.
By the continuity of $f_L$ and $f_R$ and the compactness of $B$, $M(\alpha)$ is well-defined and continuous.

Label  the critical points of $f_L$ in $[a_1,a_2]$ as  $\alpha^L_1 < \alpha^L_2 < \cdots < \alpha_{t-1}^L$ and set  $\alpha^L_0=a_1$ and  $\alpha^L_t=a_2$.
The intervals $I^L_k =[\alpha^L_{k-1}, \alpha^L_{k}]$ partition  $[a_1,a_2]$
 (note that the subintervals overlap at the critical points).  
Let $m^L_k, \beta^L_k$ be such that $$\forall \alpha \in I^L_k, \quad f_L(\alpha) = m^L_k \alpha + \beta^L_k.$$ By the conditions of the  lemma,
$m^L_1 < m^L_2 < \cdots   < m^L_t$.

Similarly,
label  the critical points of $f_R$ in $[b_1,b_2]$ as  $\alpha^R_1 < \alpha^R_2 < \cdots < \alpha_{u-1}^R$ and set  $\alpha^R_0=b_1$ and  $\alpha^R_u=b_2$.
The intervals $I^R_s =[\alpha^R_{s-1}, \alpha^R_{s}]$ similarly partition $[b_1,b_2]$. 
Let $m^R_s, \beta^R_s$ be such that $$\forall \alpha \in I^R_s, \quad f_R(\alpha) = m^R_s \alpha + \beta^R_s.$$ By the conditions of the  lemma,
$m^R_1 < m^R_2 < \cdots   < m^R_u$.

Finally, let  $\breve{I}$ denote the largest open interval contained in $I$, so $\breve I^L_k =(\alpha^L_{k-1}, \alpha^L_{k})$ and 
$\breve I^R_s =(\alpha^R_{s-1}, \alpha^R_{s}).$

\medskip

Now define
$$ M(y,\alpha_1,\alpha_2) = \max\left(  f_L(\alpha_1) + y,  f_R(\alpha_2) - y\right).$$

For fixed   $\alpha$, further 
define
$$
T(\alpha) =
\Bigl\{
(y,\alpha_1,\alpha_2)\,:\ 
y \in [\ell,r,], \ (\alpha_1,\alpha_2) \in B(\alpha),\,
\mbox{ and }
 M(y,\alpha_1,\alpha_2)  = M(\alpha)
\Bigr\}.
$$
Every 
$(y, \alpha_1,\alpha_2) \in T(\alpha)$, 
is called a  {\em candidate triple} (for $\alpha$).

%

Now consider  the following 
 eight  conditions:
$$\begin{array}{llcll}
(C_1)  &  y  = \ell,                                                                              & \quad & (C_2) &   y = r,\\
(C_3) &  y \not\in \{\ell, r\}  \mbox{ and }  \alpha_1 = a_1,  & \quad & (C_4) &  y \not\in \{\ell, r\}  \mbox{ and }  \alpha_1 = a_2,\\
(C_5) &  y \not\in \{\ell, r\}  \mbox{ and }  \alpha_2 = b_1,  & \quad & (C_6) &  y \not\in \{\ell, r\}  \mbox{ and }  \alpha_2 = b_2,
\end{array}
$$
\par\noindent $(C7)$  At least one of  $\alpha_1 = \alpha_k^L$ or  $\alpha_2 = \alpha_s^R$  is true for some $k$ or $s$\\ \hspace*{.8in} and
$(y,\alpha_1,\alpha_2)$  does not satisfy $(C_1)-(C_6)$,\\[0.03in]
 $(C_8)$ None of $(C_1)-(C_7)$ is satisfied.

\bigskip

Similar to the proof of  \Cref{lem: fork M lem} a function $\bar M(\alpha)$ is called a {\em witness for condition $(C_i)$} if
$$
\bar M(\alpha) 
\left\{
\begin{array}{ll}
=M(\alpha)  & \mbox{if there exists  $(y,\alpha_1,\alpha_2) \in T(\alpha)$ that satisfies condition $C_i$},\\
\ge M(\alpha) & \mbox{if there does not  exist any  $(y,\alpha_1,\alpha_2) \in T(\alpha)$ that satisfies condition $C_i$}.\\
\end{array}
\right.
$$
By default,  if no  
$(y,\alpha_1,\alpha_2) \in T(\alpha)$ exists that satisfies condition $(C_i)$ and $\bar M(\alpha)$ is undefined, 
 we will assume that $\bar M(\alpha) = \infty$ (so that it is trivially a witness).
 
\medskip

Every candidate triple  must  satisfy at least one of  conditions $(C_1)$-$(C_8)$;   Claim 7  later will show that $(C_8)$ is superfluous and that every $\alpha$ will be witnessed 
by some $(C_i)$,  $i \le 7.$  Similar to the proof of \Cref{lem: fork H lem}, we construct  $O(n)$-size piecewise linear witness functions for each  $(C_i)$,  $i \le 7,$  and then take their minimum.  The main work will be for $(C_7)$, where it is not a-priori obvious that the witness  function has size $O(n).$

\medskip

Set
\begin{eqnarray}
\hspace*{-.4in} M^{(1)}:  [a_1+b_1, a_2+b_2]  \rightarrow \Re, &\ &     M^{(1)}(\alpha) = \min_{ (\alpha_1,\alpha_2) \in B(\alpha)} M(\ell,\alpha_1,\alpha_2),  \label{eq:fork 10}\\
M^{(2)}:  [a_1+b_1, a_2+b_2]  \rightarrow \Re, &\ &     M^{(2}(\alpha) = \min_{ (\alpha_1,\alpha_2) \in B(\alpha)}  M(r,\alpha_1,\alpha_2). \label{eq:fork 11}
 \end{eqnarray}
Direct application of  \Cref{lem: fork good cons,lem: fork M lem} yields that both $ M^{(1)}$ and $ M^{(2)}$ are good functions that can be constructed in $O(n)$ time
 and are, respectively, witnesses for conditions $(C_1)$ and $(C_2).$

Now define
$$y(\alpha_1,\alpha_2) =  \frac 1 2 \left(f_R(\alpha_2) - f_L(\alpha_1)\right).$$
Note that 
$$
 f_L(\alpha_1) + y(\alpha_1,\alpha_2) = \frac 1 2 \left(f_R(\alpha_2) + f_L(\alpha_1)\right)\\
 								    = f_R(\alpha_2) - y(\alpha_1,\alpha_2)								  
$$
so
\begin{equation}
\label{eq: fork 70}
M(y(\alpha_1,\alpha_2),\alpha_1,\alpha_2)  = \frac 1 2 \left(f_R(\alpha_2) + f_L(\alpha_1)\right).
\end{equation}

\par\noindent\underline{Claim 1:} If  $(y,\alpha_1,\alpha_2)\in T(\alpha) $,   then
\begin{equation}
\label{eq:fork 7}
y =
\left\{
\begin{array}{ll}
l  & \mbox{ if $y(\alpha_1,\alpha_2) \le l$}\\
y(\alpha_1,\alpha_2) & \mbox{ if  $l \le  y(\alpha_1,\alpha_2) \le  r$}\\
r &  \mbox{ if $y(\alpha_1,\alpha_2) \ge r$}
\end{array}
\right.
\end{equation}

\medskip

There are three cases to check.
 								    
    \medskip
 								    								    
\par\noindent Case (a):  Assume  $l \le y(\alpha_1,\alpha_2)\le r$:  
If $y > y(\alpha_1,\alpha_2)$, then 
$$ M(y,\alpha_1,\alpha_2)  \ge  f_L(\alpha_1) + y > M(y(\alpha_1,\alpha_2),\alpha_1,\alpha_2).$$
Similarly, 
if $y < y(\alpha_1,\alpha_2)$, then 
$$ M(y,\alpha_1,\alpha_2)  \ge  f_R(\alpha_2) - y >   M(y(\alpha_1,\alpha_2),\alpha_1,\alpha_2).$$
Thus, $(y,\alpha_1,\alpha_2)\in T(\alpha) $ implies $y = y(\alpha_1,\alpha_2).$

\medskip

\par\noindent Case (b): Assume $ y(\alpha_1,\alpha_2) \le l:$
If $l < y$, then
\begin{eqnarray*}
M(y,\alpha_1,\alpha_2) 
&\ge&    f_L(\alpha_1) + y\\
& >&   f_L(\alpha_1) + l\\
& \ge & f_L(\alpha_1) +y(\alpha_1,\alpha_2)\\
&=& f_R(\alpha_2) - y(\alpha_1,\alpha_2) \\
&\ge& f_R(\alpha_2) -l.
\end{eqnarray*}
Thus
$$
M(y,\alpha_1,\alpha_2)  > \max  
 \left(
f_L(\alpha_1) + l,\,   f_R(\alpha_2) -l   
 \right)
 = 
M(l, \alpha_1,\alpha_2).
$$
So $(y,\alpha_1,\alpha_2)\in T(\alpha) $ implies $y = l.$

\medskip

\par\noindent Case (c): Assume  $   y(\alpha_1,\alpha_2) \ge r:$
If $y < r$ then 
\begin{eqnarray*}
M(y,\alpha_1,\alpha_2) 
&\ge&    f_R(\alpha_2) - y\\
& >&   f_R(\alpha_2) - r \\
& \ge & f_R(\alpha_2) -y(\alpha_1,\alpha_2)\\
&=& f_L(\alpha_1) + y(\alpha_1,\alpha_2) \\
&\ge& f_L(\alpha_1) +r.
\end{eqnarray*}
Thus
$$
M(y,\alpha_1,\alpha_2)  > \max  
 \left(
f_L(\alpha_1) + r,\,   f_R(\alpha_2) -r  
 \right)
 = 
M(r, \alpha_1,\alpha_2).
$$
So $(y,\alpha_1,\alpha_2)\in T(\alpha) $ implies $y = r.$

This completes the proof of Claim 1.

%

\medskip
\par\noindent\underline{Claim 2:} If   $(y,\alpha_1,\alpha_2)\in T(\alpha)$ 
and  $y \not\in \{\ell, r\}$ 
then 
$$
y =y(\alpha_1,\alpha_2) 
\quad\mbox{and}\quad 
M(y,\alpha_1,\alpha_2) =  \frac 1 2 \left( f_L(\alpha_1) + f_R(\alpha_2)  \right).$$
Claim 2 follows directly from Claim 1 and \Cref{eq: fork 70}.

\medskip

Now define (for the appropriate ranges)
$$
\begin{array}{ccccc}
M^{(3)}(\alpha)  &=& M\left(y(a_1, \alpha - a_1), a_1, \alpha - a_1\right) &=& \frac 1 2 \left( f_L(a_1) + f_R(\alpha - a_1)\right).\\
M^{(4)}(\alpha)  &=& M\left(y(a_2, \alpha - a_2), a_2, \alpha - a_2\right) &=& \frac 1 2 \left( f_L(a_2) + f_R(\alpha - a_2)\right).\\
M^{(5)}(\alpha)  &=& M\left(y(\alpha - b_1,b_1),   \alpha - b_1, b_1\right) &=& \frac 1 2 \left( f_L(\alpha -b_1) + f_R(b_1)\right).\\
M^{(6)}(\alpha)  &=& M\left(y(\alpha - b_2,b_2),   \alpha - b_2, b_2\right) &=& \frac 1 2 \left( f_L(\alpha - b_2) + f_R(b_2)\right).
\end{array}
$$
Multiple applications of \Cref{lem: fork good cons}   show that $M^{(i)}$ $i=3,4,5,6$ are all  positive good functions that can be constructed in $O(n)$ time.

From Claims 1 and  2,  
 if  $(y,\alpha_1,\alpha_2)\in T(\alpha)$ satisfies condition $(C_i)$ for $i=3,4,5,6$,  then $M (\alpha) = M^{(i)}(\alpha)$.
Thus $M^{(i)}$ $i=3,4,5,6$ are witnesses for condition $(C_i).$

%
%
%

\medskip

We now consider when a candidate triple $(y,\alpha_1,\alpha_2) \in T(\alpha)$  satisfies condition $(C_7).$
We first derive properties that will permit constructing this witness function quickly.
%
%

\medskip
\par\noindent\underline{Claim 3:}  Suppose $(y,\alpha_1,\alpha_2) \in T(\alpha)$  does not satisfy any of conditions  $(C_1)-(C_6)$
\begin{itemize}
\item[(A)] Further suppose that  
$\alpha_1 = \alpha_k^L$ is a critical point of $f_L$  and $\alpha_2 \in I^R_s=\left[\alpha^R_{s-1}, \alpha^R_{s}\right]$.  Then
\begin{itemize}
\item[(i)] if $ \alpha_2 > \alpha^R_{s-1}$, then $m^R_s \le m^L_{k+1}$
\item[(ii)]  if $ \alpha_2 < \alpha^R_{s}$, then   $m^L_k  \le  m^R_s$.
\item[(iii)] If $\alpha_2 \in \breve  I^R_s$ then  $m^L_k \le m^R_s \le  m^L_{k+1}.$
\end{itemize}
\item[(B)] Now, further suppose $(y,\alpha_1,\alpha_2) \in T(\alpha)$  does not satisfy any of conditions  $(C_1)-(C_6)$, 
$\alpha_2 = \alpha_s^R$ is a critical point of $f_R$  and $\alpha_1 \in I^L_k =[\alpha^L_{k-1}, \alpha^L_{k}]$.  Then
\begin{itemize}
\item[(i)] if  $\alpha_1 > \alpha^L_{k-1},$ then $m^L_k \le m^R_{s+1}$.
\item[(ii)]   if  $\alpha_1 < \alpha^L_{k},$ then   $m^R_s \le m^L_k$.
\item[(iii)]   If $\alpha_1 \in \breve  I^L_k$ then  $m^R_s \le m^L_k \le  m^R_{s+1}.$
\end{itemize}
\end{itemize}

We prove (A).  The proof of (B) is symmetric.

From Claim 2 and $(y,\alpha_1,\alpha_2)$  not satisfying  $(C_1)$ and  $(C_2),$  we have $y = y (\alpha_1,\alpha_2)$ and  $\ell < y < r.$
From  $(y,\alpha_1,\alpha_2)$ not satisfying any of  $(C_3)-(C_6),$
we can find 
 $\epsilon$ small enough that
 $$ \forall \alpha'_1 \in [\alpha_1-\epsilon, \alpha_1+ \epsilon],\ 
\forall \alpha'_2 \in [\alpha_2-\epsilon, \alpha_2+ \epsilon],\quad
\ell <  y (\alpha'_1,\alpha'_2) < r.
$$

Thus, from Claims 1 and 2, 
\begin{eqnarray*}
M(y,\alpha_1+\epsilon,\alpha_2-\epsilon) &=&  \frac 1 2 \left( f_L(\alpha_1+\epsilon) + f_R(\alpha_2-\epsilon)  \right),\\
M(y,\alpha_1-\epsilon,\alpha_2+\epsilon) &=&  \frac 1 2 \left( f_L(\alpha_1-\epsilon) + f_R(\alpha_2+\epsilon)  \right).
\end{eqnarray*}

From $\alpha_1 = \alpha_k^L$, for small enough $ \epsilon,$
$$f_L(\alpha_1 - \epsilon)  = f_L(\alpha_1) - m^L_k \epsilon,
\quad\mbox{and}\quad
f_L(\alpha_1 + \epsilon)  = f_L(\alpha_1) + m^L_{k+1} \epsilon,
$$

To see (i), note that if if $ \alpha_2 > \alpha^R_{s-1}$, then for small enough $\epsilon > 0,$   
$$f_R(\alpha_2 - \epsilon)  = f_R(\alpha_2) - m^R_s \epsilon.$$
This implies that if  $m^R_s > m^L_{k+1},$ 
\begin{eqnarray*}
M(y,\alpha_1+\epsilon,\alpha_2-\epsilon) &=&  \frac 1 2 \left( f_L(\alpha_1+\epsilon) + f_R(\alpha_2-\epsilon)  \right),\\
 &=&  \frac 1 2 \left( f_L(\alpha_1) + f_R(\alpha_2)\right)   + \frac \epsilon 2 \left(m^L_{k+1} - m^R_s \right)\\
 &<& \frac 1 2 \left( f_L(\alpha_1) + f_R(\alpha_2)\right)  = M(y,\alpha_1,\alpha_2)
\end{eqnarray*}
contradicting the fact that $(y,\alpha_1,\alpha_2)$ is a candidate triple.  Thus,  $m^R_s \le  m^L_{k+1}$.

To see (ii), note that if  $ \alpha_2 < \alpha^R_{s}$,  then for small enough $\epsilon > 0,$ 
$$f_R(\alpha_2 + \epsilon)  = f_R(\alpha_2) + m^R_s \epsilon.$$
This implies that if $m^R_s < m^L_{k}$,
\begin{eqnarray*}
M(y,\alpha_1-\epsilon,\alpha_2+\epsilon) &=&  \frac 1 2 \left( f_L(\alpha_1-\epsilon) + f_R(\alpha_2+\epsilon)  \right),\\
 &=&  \frac 1 2 \left( f_L(\alpha_1) + f_R(\alpha_2)\right)   + \frac \epsilon 2 \left(m^R_s - m^L_k \right)\\
 &<& \frac 1 2 \left( f_L(\alpha_1) + f_R(\alpha_2)\right)  = M(y,\alpha_1,\alpha_2)
\end{eqnarray*}
again contradicting the fact that $(y,\alpha_1,\alpha_2)$ is a candidate triple.  Thus  $m^L_{k} \le m^R_s.$ 

Point (iii) follows from combining points (i) and (ii),  completing the proof of Claim 3(A).

\medskip
\par\noindent\underline{Claim 4:}  Suppose $(y,\alpha_1,\alpha_2) \in T(\alpha)$    satisfies condition $(C_7)$.
For each $\alpha_k^L$  and   $\alpha_s^R,$ define
$${\cal R}^L_k =\left\{s: m^L_k \le m^R_s \le  m^L_{k+1}\right\},\quad
{\cal R}^R_s =\left\{k: m^R_s \le m^L_k \le  m^R_{s+1}\right\}
$$
and
$${\cal I}^L_k = \bigcup_{s \in{\cal R}^L_k } I_s^R, \quad
{\cal I}^R_s = \bigcup_{k \in{\cal R}^R_s } I_k^L.
$$
Then
\begin{itemize}
\item[(A)] If $\alpha_1 = \alpha_k^L$ is a critical point of $f_L$  then \\
either  $\alpha_2 \in {\cal I}^L_k$  or   \quad   $\alpha_2=  \alpha_s^R$  (a critical point of $f_R$) and  $\alpha_1 \in {\cal I}^R_s.$
\item[(B)]  If $\alpha_2 = \alpha_s^R$ is a critical point of $f_R$  \\
then either $\alpha_1 \in {\cal I}^R_s$ or \quad $\alpha_1=  \alpha_k^L$  (a critical point of $f_L$) and  $\alpha_2 \in {\cal I}^R_k.$
\end{itemize}

We prove (A).  The proof of (B) is symmetric.  

Suppose that $(y,\alpha_1,\alpha_2) \in T(\alpha)$    satisfies condition $(C_7)$
with  $\alpha_1 = \alpha_k^L$.
If $\alpha_2 \in \breve  I^R_s$, then by Claim 3(A)(iii),  $I^R_s \subseteq {\cal I}^L_k$ and the claim is correct.

Otherwise, $\alpha_2$ is a critical point, i.e.,~ $\alpha_2 = \alpha^R_s$ for some $s$. Claims 3(A) (i) and (ii) then imply
$$
m^R_s \le  m^L_{k+1}
\quad\mbox{and} \quad 
m^R_{s+1} \ge m^L_k.
$$
If $m^L_k \le m^R_s$ then $\alpha_2 \in I^R_s \subseteq {\cal I}^L_k$ and the claim is correct. Otherwise  $m^R_s < m^L_k \le m^R_{s+1}$ and
 $\alpha_1 \in I^L_k \subseteq {\cal I}^R_s$ and the claim is correct.

\medskip

The decomposition above permits  constructing a  compact function to witness condition $(C_7)$.  It will match each critical point of $f_L(\alpha)$  ($f_R(\alpha)$) to its associated interval in $f_R(\alpha)$ ($f_L(\alpha)$).

\begin{figure}
\centerline{\includegraphics[width = 5in]{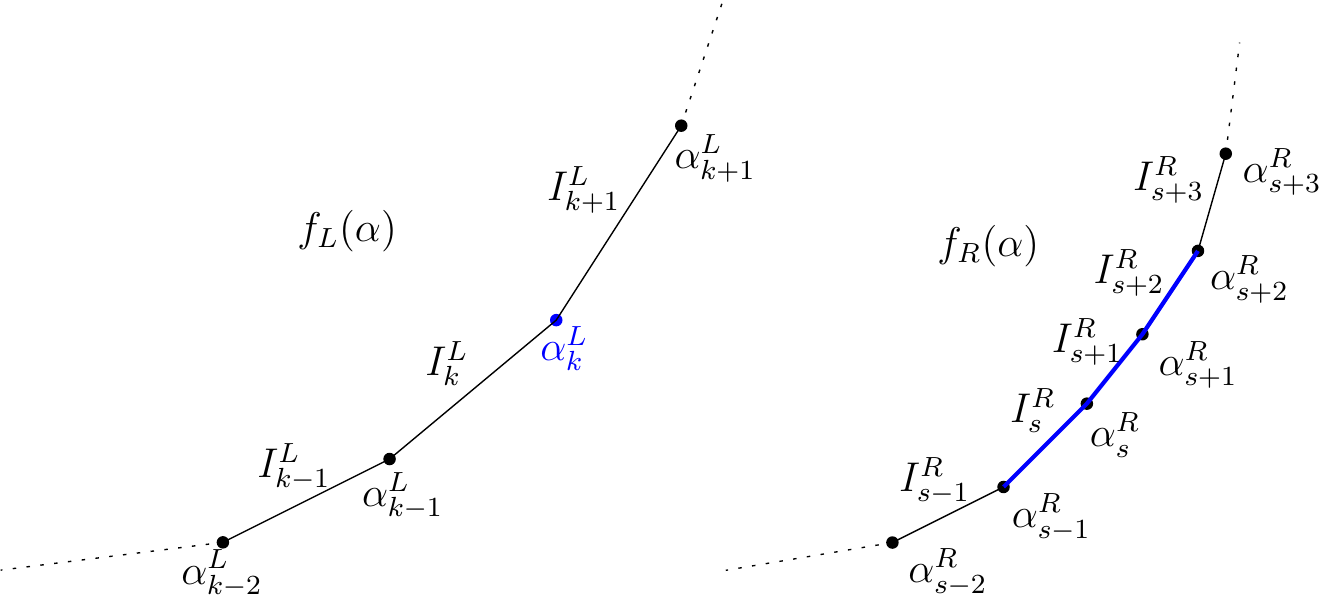}}
\caption{Illustration of situation in Claim 4 with the two functions $f_L(\alpha)$, $f_R(\alpha)$  and their critical points.  Note that the slopes of of $f_R(\alpha)$ in intervals  $I^R_s,I^R_{s+1},I^R_{s+2}$ fall between the slopes of $I^L_k,I^L_{k+1}$  so  ${\cal R}^L_k = \{s,s+1,s+2\}.$}
\label{fig:critvert}
\end{figure}

\medskip

Before continuing, we introduce the following useful notation.
If $I=[a,b]$ let $\alpha + I$ denote the interval $[a+\alpha,b+\alpha].$  We say that $[a_1,b_1] = I_1 < I_2 =[a_2,b_2]$ if $b_1 < a_2$  and 
 $I_1 \le  I_2$ if  $b_1 \le a_2.$

It would be elegantly convenient for the later proof if the $\alpha^L_k + I^R_s$ with  $s \in {\cal R}^L_{k}$ were all disjoint.  Unfortunately, this is not true.  The best that can be proven is the next claim (which suffices for our purposes).

\medskip
\par\noindent\underline{Claim 5:} 
\begin{itemize}
\item[(A)]  Let $k_1 < k_2$  and $s_1 \in {\cal R}^L_{k_1}$  and $s_2 \in {\cal R}^L_{k_2}$ such that 
 $ I = (\alpha_{k_1}^L + I_{s_1}^R) \cap (\alpha_{k_2}^L + I_{s_2}^R) \not = \phi.$\\  Then
$ k_2=k_1+1.$
\item[(B)]  Let $s_1 < s_2$  and $k_1 \in {\cal R}^R_{s_1}$  and $k_2 \in {\cal R}^R_{s_2}$ such that 
 $ I = (\alpha_{s_1}^R + I_{k_1}^L) \cap (\alpha_{s_2}^R + I_{k_2}^L) \not = \phi.$ \\
  Then
  $s_2=s_1+1$.
\end{itemize}

We prove (A).  The proof of (B) is symmetric.

Since  $m_{k_1+1}^L \leq m_{k_2}^L$, by definition  $s_1 \le s_2$.  If $s_1 < s_2$, then 
 $I_{s_1}^R \le I_{s_2}^R$, and as $\alpha_{k_1}^L < \alpha_{k_2}^L$, we get  $\alpha_{k_1}^L + I_{s_1}^R < \alpha_{k_2}^L + I_{s_2}^R$ so $I = \emptyset.$  Thus,  we may assume $s_1=s_2.$

 From the definition of   ${\cal R}^L_k$, 
$m^R_{s_1} \le m^L_{k_1+1} \le m^L_{k_2} \le m^R_{s_2}$. Since $s_1 = s_2,$  $k_2=k_1+1.$

\medskip

\par\noindent\underline{Claim 6:}   
Define
\begin{eqnarray*}
h^L_k(\alpha) =
M\left(y(\alpha_k^L, \alpha - \alpha_k^L), \alpha_k^L, \alpha - \alpha_k^L\right) = \frac 1 2 \left( f_L(\alpha_k^L) + f_R(\alpha -\alpha_k^L)\right) &  \mbox{if }  \  \alpha \in \alpha_k^L +  {\cal I}^L_k,\\
h^R_s(\alpha) = 
M\left(y(\alpha - \alpha_s^R,\alpha_s^R),\alpha - \alpha_s^R,\alpha_s^R\right) = \frac 1 2 \left( f_L(\alpha- \alpha_s^R) + f_R(\alpha_s^R)\right)   & \mbox{if }  \  \alpha \in  \alpha_s^R +  {\cal I}^R_s,
\end{eqnarray*}
where the functions have value $\infty$ outside their specified domains.
Further set
$$
H^L(\alpha) = \min_{1 \le k < t} h^L_k(\alpha),
\quad
H^R(\alpha) = \min_{1 \le s< u} h^R_s(\alpha)
\quad\mbox{and}\quad
M^{(7)}(\alpha)=
\min\left(
H^L(\alpha),\,
H^R(\alpha) 
\right).
$$

Then $M^{(7)}(\alpha)$ is a  piecewise linear function with positive slopes that is a witness to condition $(C_7)$ and can be built in $O(n)$ time.

\bigskip

To prove this claim, first note that by definition, $\forall \alpha,$  $M^{(7)}(\alpha) \ge M(\alpha)$.
Now, let   $(y,\alpha_1,\alpha_2)$ be a candidate triple for $\alpha$ that satisfies condition $(C_7)$.  This implies that either 
$\alpha_1 = \alpha^L_k$ for some $k,$  $\alpha_2 = \alpha^R_s$ for some $s$ or both at once.

Assume that $\alpha_1 = \alpha^L_k$ and set $\alpha_2 = \alpha - \alpha_1.$  From Claim 4 and Claim 2, one of the following two events  must occur.
\begin{itemize}
\item  $\alpha_2 \in {\cal I}^R_s $, and thus
$h^L_k(\alpha) = M(y(\alpha_1,\alpha_2),\alpha_1,\alpha_2) = M(\alpha)$ or
\item 
$\alpha_2=  \alpha_s^R$ is a critical point of $f_R$ and  $\alpha_1 \in {\cal I}^R_s$
and thus $h_s^R(\alpha) = M(y(\alpha_1,\alpha_2), \alpha_1,\alpha_2) = M(\alpha)$
\end{itemize}
Thus, $M^{(7)}(\alpha) = M(\alpha)$.
The proof that $M^{(7)}(\alpha) = M(\alpha)$ if we assume   $\alpha_2 = \alpha^R_s$ is symmetric.
This proves that   $M^{(7)}(\alpha)$ is a witness to condition $(C_7)$.

\medskip
We now show how to construct  $H^L(\alpha)$ in $O(n)$ time (and that it has $O(n)$ critical points). 

\medskip

To start, set $I_k = \alpha_k^L +  {\cal I}^L_k,$  the domain of  $ h^L_k(\alpha).$  
 Let  $\ell_k,r_k$ be such that $I_k = [\ell_k, r_k]$. 
Note that $I_k$ contains $|{\cal R}^L_k|+1$ critical points.
so
each  $ h^L_k(\alpha)$ can be built in time $O(|{\cal R}^L_k|+1).$   Furthermore,  since a critical point of  $H^L(\alpha)$ must be  a critical point of one of the $I_k$,  $H^L(\alpha)$ has  $O\left(  \sum_k (|{\cal R}^L_k|+1)\right)=O(n)$ critical points.

To continue, note that by construction,  $\ell_k \le \ell_{k+1},$  $r_k  \le r_{k+1}$ and, from Claim 5,  $r_{k-1} < \ell_{k+1}.$

\medskip

For $j < t$, define  
$H_j^L(\alpha) = \min\limits_{1 \le k \le j} h^L_k(\alpha).$
Its domain is  $D_j = \cup_{1 \le k \le j} I_j = D_{j-1} \cup [\ell_j,r_j]$.  

\medskip

Set $H_1^L(\alpha) =h^L_1(\alpha).$   Now assume that $H_{j}^L$ is known. We will build  $H_{j+1}^L$ from  $H_{j}^L$.
There are two possible cases.

In the first case, $r_j < \ell_{j+1}.$  Then $D_j  < I_{j+1}$ so we can just concatenate $I_{j+1}$ 
(along with the associated function information of $h^L_{j+1}(\alpha)$) to the end of $D_j$. This takes $O(|{\cal R}^L_{j+1}|+1)$ time.

In the second case,  $r_j\ge \ell_{j+1}.$  First note that,  since $r_{k-1} < \ell_{k+1}$,  $D_{j-1} \cap I_{j+1} = \emptyset.$ 
 Thus $D_j \cap I_{j+1} = I_j \cap I_{J+1}= [\ell_{j+1},r_{j}]$ and 
$$
H_{j+1}^L(\alpha) =
\begin{cases}
H_{j}^L (\alpha) &  \mbox{if  $ \alpha \in D_{j-1} \cup [\ell_{j},\ell_{j+1})$}\\
\min(h_{j}(\alpha),h_{j+1}(\alpha)) & \mbox{if  $ \alpha \in I_{j} \cap I_{j+1}= [\ell_{j+1},r_{j}]$} \\
h_{j+1}(\alpha) &\mbox{if  $ \alpha \in (r_j,r_{j+1}]$}
\end{cases}
$$
Since  $D_{j-1} \cap I_{j+1} = \emptyset,$  the critical points that define  $D_j$ in $[\ell_{j+1},r_j]$ are all critical points in $I_j$.
Thus $H_{j}^L (\alpha)$ defined on $D_j$ can be trimmed back to only being defined on $D_{j-1} \cup [\ell_{j},\ell_{j+1})$ in  $O(|{\cal R}^L_{j}|+1)$ time.  This yields $H_{j+1}^L(\alpha)$   defined on $D_{j-1} \cup [\ell_{j},\ell_{j+1})$ .

$H_{j+1}^L(\alpha)$ defined on $[\ell_{j+1},r_{j}]$ can be constructed in $O(|{\cal R}^L_{j}|+|{\cal R}^L_{j+1}|+1)$ time.

$H_{j+1}^L(\alpha)=h_{j+1}(\alpha)$ defined on $(r_j,r_{j+1}]$ can  be   constructed in $O(|{\cal R}^L_{j+1}|+1)$ time.

Concatenating the three pieces (which only intersect at their endpoints) requires only $O(1)$ more time and yields the full description of  $H_{j+1}^L(\alpha)$.

We have therefore just shown that, in both cases, the time required to construct $H_{j+1}^L(\alpha)$  from $H_{j}^L(\alpha)$ is  $O(|{\cal R}^L_{j}|+|{\cal R}^L_{j+1}|+1).$

Thus, the total time to construct $H^L(\alpha) = H_{t-1}^L(\alpha)$ is  
$$O\left( 
\sum_{1 \le k < t} 
	\left(
	|{\cal R}^L_{j}|+|{\cal R}^L_{j+1}|+1
	\right) 
\right)
= O(t) = O(n).
$$

A similar argument shows that  $H^R(\alpha)$ can also be built in $O(n)$ time and has $O(n)$ critical points.
Thus, the piecewise linear function
$M^{(7)}(\alpha)
$
with its critical points in sorted order can be built in $O(n)$ time.  Since  both $H^L(\alpha)$ and $H^R(\alpha)$ have $O(n)$ critical points, $M^{(7)}(\alpha)$ does as well.

Furthermore, since all the individual $h_k^L(\alpha)$ and  $h_s^R(\alpha)$ have positive slopes, $M^{(7)}(\alpha)$ does as well.
This completes the proof of  Claim 6.

\medskip

\medskip
\par\noindent\underline{Claim 7:}  
 Suppose $(y,\alpha_1,\alpha_2) \in T(\alpha)$ satisfies  condition  $(C_8)$.  Then there exists another  
 $(y',\alpha'_1,\alpha'_2) \in T(\alpha)$
 that satisfies  at least one of the conditions  $(C_1)-(C_7)$.

\medskip
Because $(y,\alpha_1,\alpha_2)$ does not satisfy 
 any of  conditions $(C_1)-(C_7)$, 
 $\alpha_1$ is not a critical point of $f_L$,  $\alpha_2$ is not a critical point of $f_R$ and $l < y(\alpha_1, \alpha_2) < r$ since otherwise, either $(l,\alpha_1, \alpha_2)$ or $(r,\alpha_1, \alpha_2) \in T(\alpha)$ and we are done. 
 Furthermore, there exist $k,s$ 
 and $\epsilon > 0$ such that 
\begin{equation}
\label{eq:frac: claim 4}
 [\alpha_1 - \epsilon, \alpha_1 +\epsilon] \subseteq  I^L_k,\,
\quad\mbox{and}\quad
 [\alpha_2 - \epsilon, \alpha_2 +\epsilon] \subseteq  I^R_s.
\end{equation}
and at least one of the following facts  is true:\\
(a) 
at least one of the two points  $\alpha_1 \pm \epsilon$ is  a critical point of $f_L$
(b) $\alpha_1 - \epsilon= a_1,$
(c) $\alpha_1 + \epsilon= a_2,$
(d) 
at least one of the two points  $\alpha_2 \pm \epsilon$ is  a critical point of $f_R$
(e) $\alpha_2 - \epsilon= b_1,$
(f) $\alpha_2+ \epsilon= b_2.$

Recall that \Cref{eq:frac: claim 4} implies
$$\forall \alpha_1' \in  [\alpha_1 - \epsilon, \alpha_1 +\epsilon],\, 
f_L(\alpha_1') = m^L_k \alpha_1' + b^L_k
\quad\mbox{and}\quad
\forall \alpha_2' \in  [\alpha_2 - \epsilon, \alpha_2 +\epsilon],\, 
f_R(\alpha_2') = m^R_s \alpha_2' + b^R_s.
$$

Suppose that $m^L_k\not=m^R_s$.  Without loss of generality, assume that $m^L_k > m^R_s$. Because neither condition $(C_2)$, $(C_3)$ hold, we can choose $\epsilon' < \epsilon$ small enough that 
$$y(\alpha_1 - \epsilon', \alpha_2+ \epsilon') = y(\alpha_1,\alpha_2) - \frac {\epsilon'} 2  (m^L_k - m^R_s)$$
satisfies  
$$l \le y(\alpha_1 - \epsilon', \alpha_2+ \epsilon') \le r.$$
But then
$$M(y(\alpha_1 - \epsilon', \alpha_2+ \epsilon'), \alpha_1 - \epsilon', \alpha_2+ \epsilon')
= M(y(\alpha_1,\alpha_2), \alpha_1,\alpha_2) - \frac {\epsilon'} 2 (m^L_k - m^R_s) < M(y(\alpha_1,\alpha_2), \alpha_1,\alpha_2)  =  M(\alpha) $$
contradicting the definition of  $T(\alpha).$ So this case can not occur and  $m^L_k=m^R_s$.

Now $m^L_k=m^R_s$ implies that 
$$y(\alpha_1 - \epsilon,\alpha_2+\epsilon) = y(\alpha_1,\alpha_2) = y(\alpha_1+ \epsilon,\alpha_2-\epsilon)$$
and  
$$
M(y(\alpha_1 - \epsilon,\alpha_2+\epsilon), \alpha_1 - \epsilon,\alpha_2+\epsilon)
=
M(y, \alpha_1,\alpha_2)
=
M(y(\alpha_1 + \epsilon,\alpha_2-\epsilon), \alpha_1 + \epsilon,\alpha_2-\epsilon).
$$

Since  at least one of facts  (a)-(f) is true,  this constructs a candidate triple for  $(y(\alpha_1 + \epsilon,\alpha_2-\epsilon),\alpha_1 - \epsilon, \alpha_2+\epsilon) \in T(\alpha)$ that satisfies at least one of conditions $(C_1)-(C_7)$, thus proving  Claim 7.

\medskip
\par\noindent\underline{Claim 8:}  $M(\alpha)$ is a good function that can be constructed in $O(n)$ time.

\medskip

From Claim 7, for every $\alpha$, there exists a candidate triple for $\alpha$ that satisfies  at least one of conditions  $(C_1)-(C_7)$.
We have already seen that, for each $i=1,\ldots 7,$  
$M^{(i)}(\alpha)$ is a witness to condition $(C_i)$.
Thus  $M(\alpha) = \min_{1 \le i \le 7} M^{(i)}(\alpha).$

Each $M^{(i)}(\alpha)$, $i\le 6$
is a good function that can be built in $O(n)$ time.  From Claim 6,  $M^{(7)}(\alpha)$ is  a piecewise linear function of size $O(n)$ that can be built in $O(n)$ time (it might not be good since it might not be continuous).
Thus,
 from \Cref{lem: fork good cons}, $M(\alpha)$ is a piecewise linear  function that can be built in $O(n)$ time, all of whose slopes are positive.

 $M(\alpha)$ being   good follows from the continuity of $M$  (noted at the start  of the proof).
\end{proof}



%% file: conclusion.tex
\section{Conclusion and Possible Extensions}
\label{sec:conc}



This paper provided an $O(n^4 \log n)$ time algorithm for  solving the 1-sink location minmax regret  problem on a dynamic path network with general capacities.
To the best of our knowledge, this is  the first polynomial time algorithm for solving {\em any} sink location  minmax regret  problem with general capacities for any type of graph and any number of sinks.

While polynomial,  this running time is quite high and an obvious direction for future research would be to speed it up.  One possible approach would be to note that 
 \Cref{sec: real alg}  shows that the problem can be solved in  $O(U(n) n^2 \log n)$ time where $U(n)$ is the time required to
calculate  $G_{i,j}(x)$ and the other functions introduced in \Cref{def:gandh defs}.  The second half of the paper,  \Cref{sec: Main Utility,sec: Utility},  develop a machinery for  proving that  $U(n)=O(n^2)$.

Any improvement to $U(n)$ would improve the algorithm. 
 We  note without details that an even  more intricate analysis than that presented here {\em could} evaluate  $G_{i,j}(x)$  and $G_j(x)$  in $O(n)$ rather than $O(n^2)$ time.
This analysis uses  amortization  to show  that not only is $M^{(u)}_i(x)$ a good function of size $O(n)$ for each $u,$ but the full $M_i(x)$ has size $O(n)$  as well (the analysis presented in this paper only shows that $M_i(x)$ has size $O(n^2)$).  Straightforward modifications of \Cref{lem: Gij eval,lem: Gjeval} would then evaluate $G_{i,j}(x)$  and $G_j(x)$  in $O(n)$ time.

The reason that this approach can not (yet) be used to derive a better bound  on $U(n)$ is  that 
the  amortization  argument strongly uses the fact that $M_i(x)$  is a piecewise linear function of only one varying parameter. This permits showing that the different $M^{(u)}_i(x)$  can not {\em all} be large and thus  $M_i(x) = \max_u M^{(u)}_i(x)$ is not composed of many pieces.
 The amortization argument fails for  $M_{i,j}(x)$ though, because $M_{i,j}(x)$ is  fundamentally  a piecewise linear  function in {\em two} varying parameters. Thus, it would not be possible to use this to prove that $U(n) = O(n).$

  This failure does highlight  that a possible method of improving the algorithm would be developing a different approach to showing that 
$M_{i,j}(x)$ has size $O(n)$. An $O(n)$ time construction of  $M_{i,j}(x)$ (it is unknown whether this is possible) would immediately imply that $U(n) = O(n)$ and lead to an $O(n^3 \log n)$ time algorithm.

It would also be interesting to try to solve the $k$-sink location minmax regret problem on a dynamic path with general capacities for any $k >1.$  The corresponding algorithms 
\cite{Li2014,Bhattacharya2015,arumugam2019optimal}
 in the uniform capacity case strongly utilized the combinatorial structure of worst case solutions that were {\em independent} of the actual scenario weight values.  Because the worst case scenarios in the general capacity case {\em are} dependent on the actual weight values, those techniques can not be easily  transferred.

Similarly, it  would also be interesting to  try to solve the $1$-sink location minmax regret problem on a dynamic {\em tree} with general capacities.  The corresponding regret problem with uniform capacities \cite{Higashikawa2014g,Bhattacharya2015} used the fact that the optimization (not regret) problem on a tree with uniform capacities could be transformed to a path problem.  This reduction is no longer valid in the general capacity case and so those techniques can also not be easily transferred here.